\documentclass[nonacm,screen]{acmart}

\renewcommand\footnotetextcopyrightpermission[1]{} %
\usepackage[noend]{algpseudocode}
\usepackage[ruled]{algorithm}
\usepackage{setspace}
\usepackage{subcaption}
\usepackage{amsmath, amsthm} %
\DeclareMathOperator*{\argmax}{arg\,max}
\DeclareMathOperator*{\argmin}{arg\,min}
\usepackage{wrapfig}

\AtBeginDocument{%
  }

\setcopyright{acmlicensed}
\copyrightyear{2026}
\acmYear{2026}
\acmDOI{XXXXXXX.XXXXXXX}

\begin{document}

\title{Location-of-interest vs. Current-location Privacy in Multi-hop Communication}

\author{Ryan Vogt}
\email{rvogt@ualberta.ca}
\orcid{0009-0002-8519-7465}
\author{Ioanis Nikolaidis}
\email{nikolaidis@ualberta.ca}
\orcid{0000-0003-1469-5280}
\author{Mario A. Nascimento}

\email{mario.nascimento@ualberta.ca}
\orcid{0000-0002-7609-1805}
\affiliation{%
  \institution{University of Alberta}
  \department{Department of Computing Science}
  \city{Edmonton}
  \state{Alberta}
  \country{Canada}
}

\begin{abstract}
Consider a standard multi-hop wireless radio network in which nodes relay messages for one another. While the contents of messages may be sensitive and thus encrypted, the very fact that two radios communicated can leak sensitive information about both the physical location of the origin and the physical location of the destination. We investigate the trade-off between protecting the privacy of a user's current location and the privacy of their location of interest, and how the user should choose dummy destinations to obfuscate their queries.
\end{abstract}



\maketitle

\newcommand{\proc}[1]{\textnormal{\scshape#1}}
\newcommand{\const}[1]{\textnormal{\scshape#1}}
\newcommand{\To}{~\mathbf{to}~}
\newcommand{\Continue}{\textbf{continue}}
\newcommand{\Break}{\textbf{break}}
\newcommand{\Error}{\textbf{error}~}
\newcommand{\Preconditions}{\textbf{Preconditions:}~}
\newcommand{\fixme}[1]{\textit{\textbf{[FIXME: #1]}}}

\newcommand{\image}[4][1]{ \begin{figure} \centering \includegraphics[width=#1\textwidth]{#2} \caption{#3} \label{#4} \end{figure} }

\newcommand{\doubleimage}[5][0.5]{ \begin{figure*} \centering \subfigure{ \includegraphics[width=#1\textwidth]{#2} } \subfigure{ \includegraphics[width=#1\textwidth]{#3} } \caption{#4} \label{#5} \end{figure*} }

\newcommand{\doubleimagecap}[9][0.5]{ \begin{figure*} \centering \subfigure[#5]{ \includegraphics[width=#1\textwidth]{#2}\label{#8} } \subfigure[#6]{ \includegraphics[width=#1\textwidth]{#3}\label{#9} } \caption{#4} \label{#7} \end{figure*} }

\newcommand{\graph}[3]{\image[0.60]{#1}{#2}{#3}}

\newcommand{\doublegraph}[4]{\doubleimage[0.47]{#1}{#2}{#3}{#4}}

\section{Introduction}

We consider the problem of location privacy in multi-hop wireless communication. In multi-hop wireless communication, a transmission originates at a node, is forwarded by some number of intermediate nodes (the \textit{relays}) and is received by a destination node. Such networks form the basis of many tactical radio systems used by ground forces, where individual radios relay traffic for one another over ranges that exceed direct line-of-sight or single-hop connectivity. In these environments, the mere fact that two radios exchanged traffic can reveal the physical locations of both the originator and the destination, even when message payloads are protected by strong encryption.

Our basic assumption is that: \textit{the sequence of relayed transmissions of a message, from origin to destination, reveals a purposeful interaction between the two endpoints; that interaction maps onto the real world, and reveals a purposeful interaction between the physical locations of the two endpoints; and, the two locations and the fact that they communicated is considered sensitive information}. This assumption is warranted for environments in which an adversary can intercept wireless transmissions at multiple locations (including, in the extreme, across the entire deployment). The eavesdropping adversary is defined more precisely in Section~\ref{sect:ch2-adversary}. The key observation is that the eavesdropping scenario remains applicable under \textit{any} link-level or end-to-end encryption that could be applied to the messages. The location of the origin and the location of the destination can be inferred from the pattern of relays, independent of encryption.

We refer to messages introduced by an originator as \textit{queries}, to capture the fact that each message results in a \textit{reply} from the destination back to the origin. This model abstracts a broad category of tactical interactions, ranging from request/response exchanges for situational data or fire-support requests, to simple acknowledged command traffic.

The problem that we study is the privacy of the two \emph{locations} in question for the user: the user's \textit{current location} (the physical location of the origin) and their \textit{location of interest} (the physical location, or approximate area, of the destination). We emphasize that the problem studied is \emph{not} the privacy of the contents of the query and the reply, but instead the privacy of the physical locations of the endpoints.

Our contribution is a method inspired by $k$-anonymity~\cite{sweeney:kanon} that can be used to obfuscate the user's location of interest: the user issues not one query to their location of interest, but rather $k$ queries to a diverse set of locations. The trade-off to this approach, aside from the energy and bandwidth cost of the additional $k-1$ queries, is that the additional traffic originating from the user's current location facilitates an eavesdropping adversary determining that location. Our contribution is the study of the \textit{trade-off between protecting the privacy of the user's current location and the privacy of their location of interest, and how the user should choose the $k-1$ dummy destinations to obfuscate the real physical location of interest}.

It is important to note that in our scenario, the user trusts their own radio. However, the user does not trust the network, as there may be compromised radios. Additionally, our approach does not require that intermediate radios run any specialized software or privacy protocols. The multiple queries are initiated solely from the user's radio. In that regard, it is fundamentally different from onion-routing approaches such as Tor~\cite{dingledine:tor}. Onion routing is expensive in computational and energy cost on resource-constrained tactical radios, and critically requires a network specifically designed to support it. Our scheme works on any multi-hop radio network that can already forward ordinary messages.

We begin in Section~\ref{sect:ch2-related} by discussing related work.
We present our adversary model and assumptions in more detail in Section~\ref{sect:ch2-adversary}.
In Section~\ref{sect:ch2-metrics}, we present formal metrics for measuring the privacy of each of: the user's current location, their location of interest, and (more broadly) their area of interest.
We then discuss how a user should choose their $k-1$ dummy locations of interest to query in Section~\ref{sect:choose-fake}, followed by a discussion of the related ``partition'' and ``partitioning'' problems in Section~\ref{sect:partition-partitioning}.
We show the results of our $k$-query scheme in a featureless, uniform environment in Section~\ref{sect:ch2-results}, before adding mobility data drawn from an urban mobility simulator in Section~\ref{sect:ch2-mobility}.
In Section~\ref{sect:ch2-conclusion} we conclude and discuss future work.

\section{Related Work}\label{sect:ch2-related}

The relation between privacy and mobility in multi-hop wireless networks has been studied for some time. The literature is spread over a number of application domains, sometimes published as related to Mobile Ad-Hoc Networks (MANETs), or specific Wireless Sensor Networks (WSNs), or in more general settings.  

For example,  Hong~et~al.~\cite{hong:mobility} examined how anonymous communication in MANETs is affected by node movement, and Andersson~et~al.~\cite{andersson:requirements} outlined general requirements for privacy in MANETs. Gruteser and Grunwald~\cite{gruteser:cloaking} studied spatial cloaking so that users remain anonymous when accessing location-based services, while Gedik and Liu~\cite{gedik:k} examined how a \(k\)-anonymity approach can be realized for location-based services. Mobility patterns extracted from cell-phone records have also been shown to identify individuals~\cite{wicker:loss, zyga:phone}. The impact of compromised nodes on privacy has been investigated as well, for example by Pinto~et~al.~\cite{pinto:secrecy} which introduced the idea of routing around known compromised nodes to preserve privacy.

A significant body of work on source-location privacy in WSNs focuses on modifying the network itself. Misra and Xue~\cite{misra:clustered} show how clusters of nodes can generate and share pseudonyms when communicating with a sink. Their idea was extended by Ouyang~et~al.~\cite{ouyang:anonymity} to handle compromised shared keys. These schemes, however, address only node-to-sink traffic and are not applicable when nodes must communicate with one another to route user queries. Ozturk~et~al.~\cite{ozturk:source} and Kamat~et~al.~\cite{kamat:source} introduced phantom routing, in which each epoch’s data first follows a random directed walk away from the source before being forwarded to the sink. A similar approach is that of  Yang~et~al.~\cite{yang:global} but under a global eavesdropper. A related line of work examines an adversary that moves toward the receiver of sensor traffic~\cite{jian:receiver}. These approaches assume a mobile adversary that physically tracks a continuous stream of packets, and they require changes to the routing protocol. Subsequent evaluations have shown that phantom routing’s protection can degrade substantially under realistic network sizes and configurations~\cite{gu:phantom}.

Client-side dummy generation has long been used to achieve \(k\)-anonymity for location-based services without a trusted third party. Kido~et~al.~\cite{kido:dummy} introduced the idea of generating \(k-1\) dummy locations at the client so that a server cannot distinguish the real location. Later work refined dummy selection to improve entropy and geographic dispersion~\cite{niu:kanon, wu:dummy}. These techniques, however, assume a central server that receives the set of locations -- they do not address the path-consistency attacks that arise when the same queries must be routed across a multi-hop radio network under partial observation by an eavesdropper or compromised nodes.

Even protocols explicitly designed for anonymity remain vulnerable to traffic analysis. Liu~et~al.~\cite{liu:tia} demonstrated that a passive global adversary can recover end-to-end flows and paths in anonymous MANETs by correlating inter-arrival times. Stronger anonymity mechanisms for MANETs, such as the ALARM framework~\cite{eldefrawy:alarm}, rely on heavy cryptographic primitives and protocol changes that are costly for resource-constrained tactical radios. Onion-routing approaches~\cite{goldschlag:onion, dingledine:tor} similarly require a network specifically designed to support them and incur substantial computational and energy overhead on intermediate nodes.

The largest difference between these existing lines of work and our approach is that we focus on an environment in which the user trusts their own radio but does not trust the network (compromised radios or eavesdroppers may be present). The multiple queries are initiated solely from the user’s radio, without specialized software or privacy protocol  required at intermediate nodes. Consequently the scheme works on any multi-hop radio network that can already forward ordinary messages. This source-only design avoids the fragility of network-level routing modifications while still accounting for the path information an eavesdropper or compromised nodes can observe.

\section{Adversary Model and Assumptions}\label{sect:ch2-adversary}

We consider an adversary who can compromise individual nodes in a multi-hop wireless network. There are no physical or topological restrictions on which nodes the adversary may compromise, hence, any subset of the nodes may be under adversarial control. The adversary is incorporeal in the sense that it does not need to move through the deployment. Hence, once a node is compromised, the adversary instantly obtains all information known to that node.  Specifically, a compromised node yields to the adversary any cryptographic keys it holds and the full contents of every message it routes or overhears (whether or not those messages are encrypted under keys the adversary has obtained). 

Because the wireless medium is shared, a compromised node may also observe transmissions from honest nodes within its communication range. Interference or collisions, however, may cause a transmission to be missed, so the absence of a received message is not reliable evidence that a nearby node was uninvolved. Consequently, the adversary treats only the set of compromised nodes that did \textit{not} participate in a given query as nodes that are known \textit{not} to have routed that query.

When a user issues \(k\) queries, we assume the adversary knows the value of \(k\). This holds even if the adversary does not overhear every query. This is realistic because the entropy of a user-chosen \(k\) is small relative to the size of the network -- an adversary that observes only a subset of the queries can still bound the number of missing queries tightly.

For the purpose of analysis and simulation we adopt two simplifying assumptions about the user’s local link. First, we treat the user as interacting directly with the network node \(C\) that is closest to the user’s physical location. In practice the user would employ a handheld radio that communicates with \(C\). The additional \(2k\) transmissions on that last hop (the \(k\) queries outbound and the \(k\) replies inbound) do not alter the asymptotic privacy metrics and would only add unnecessary complexity. Second, we assume that the link between the user’s radio and \(C\) uses sufficiently low power that it is not overheard by any other node unless \(C\) itself is compromised. Both assumptions are standard in the source-location privacy literature and isolate the privacy effects that arise from multi-hop routing inside the network.

Under these conditions the adversary’s goal is to infer the user’s current location \(C\) and real location of interest \(L\) from the partial route information obtained from compromised nodes and from overheard transmissions. The precise definitions of current-location and location-of-interest privacy, and the algorithmic approach an adversary can take to attack them, are presented in the next section.

\section{Privacy Metrics}\label{sect:ch2-metrics}

We begin by defining our notation.  The network consists of a set
$\mathcal{N}$ of nodes, where \mbox{$\left|\mathcal{N}\right| = n$}.  The user
will issue $k$ queries, $Q_i$, each directed to a location $L_i$, where
\mbox{$\mathcal{L} = \left\{L_1, L_2, \ldots, L_k \right\} \subseteq
\mathcal{N}$}.  One node, \mbox{$L \in \mathcal{L}$}, is the user's real location
of interest, and the remaining \mbox{$k-1$} nodes are dummy locations of interest.
However, some of the nodes in the network may be compromised by the adversary,
and those nodes are able to overhear communication on the shared wireless medium.
Formally, the adversary attempts to
decide on the current location~(CL) and location of interest~(LOI)
of the user using the information overheard.

To study how well this $k$-anonymity scheme preserves the privacy of the user's
LOI and CL, we require formal methods to measure the privacy levels that result
from any given set of $k$ queries.  In
the following section, we define the metrics used to determine how well the
user's CL-privacy and LOI-privacy are protected.
We follow that up with a discussion of the closely related concept
of the privacy of the user's area of interest~(AOI) in Section~\ref{sect:ch2-aoi}.

\newcommand{\probGlobal}{\proc{P}_\text{\textnormal{G}}}
\newcommand{\probC}[1]{\mbox{$\probGlobal\left(C = #1\right)$}}
\newcommand{\probpC}[1]{\mbox{$\probGlobal\left(C = #1~|~C \in \mathcal{C}\right)$}}
\newcommand{\probL}[2]{\mbox{$\probGlobal\left(L = #1~|~C = #2\right)$}}
\newcommand{\probpL}[2]{\mbox{$\probGlobal\left(L = #1~|~C = #2~\text{and}~L \in \mathcal{L}_{#2}\right)$}}
\begingroup
\renewcommand*{\arraystretch}{1.4}
\begin{table}
\centering
\begin{tabular}{| l | p{3.5in} |}
\hline
$C$ & The current location~(CL) of the user, \mbox{$C \in \mathcal{N}$} \\
\hline
$C^*$ & The adversary's final guess as to the current location of the user \\
\hline
$\mathcal{C}$ & The set of possible current locations of the user, as determined by the adversary \\
\hline
$k$ & The total number of queries sent by the user, where \mbox{$k-1$} are to dummy locations of interest and
     one is to the real LOI\\
\hline
$L$ & The user's real location of interest~(LOI) \\
\hline
$L^*$ & The adversary's final guess as to the user's real location of interest \\
\hline
$\mathcal{K}$ & The set of query destinations known to the adversary, i.e.,
     \mbox{$\mathcal{K} \subseteq \mathcal{L} $} \\
\hline
$\mathcal{L} = \left\{L_1, \ldots, L_k\right\}$ & The $k$ locations queried by the user, with \mbox{$L \in \mathcal{L}$} \\
\hline
$\mathcal{L}_N$ & The set of possible locations queried by the user, as determined by the adversary,
     under the assumption that \mbox{$C = N$} for some \mbox{$N \in \mathcal{C}$} \\
\hline
$\proc{M}_\mathit{AOI},~\proc{M}_\mathit{CL},~\proc{M}_\mathit{LOI}$ & The measures of the user's AOI-, CL-, and
     LOI-privacy respectively \\
\hline
$n$ & The number of nodes in the network \\
\hline
$\mathcal{N} = \left\{N_1, \ldots, N_n\right\}$ & The nodes in the network, with \mbox{$|\mathcal{N}| = n$} \\
\hline
$\probGlobal\left(\text{condition}\right)$ & A global, aggregate probability over all users that the given condition is true;
     for example, \probpC{N} is the probability that any user's current location is $N$, given that their current location
     is known to be in $\mathcal{C}$ \\
\hline
$\mathcal{Q} = \left\{Q_1, \ldots, Q_k\right\}$ & The $k$ queries sent by the user \\
\hline
$\mathcal{R}_i = \left(N_{i, 1}, \ldots, N_{i, l_i}\right)$ & The route taken by query
     $Q_i$, transmitted through $l_i$ nodes \\
\hline
$\sigma\left(\mathcal{L}, \mathcal{N}\right)$ & A measure of how physically dispersed the elements of
     $\mathcal{L}$ are over the network $\mathcal{N}$ \\
\hline
$\mathcal{U}$ & The set of queries with destinations unknown to the adversary, i.e.,
     \mbox{$\mathcal{U} \subseteq \mathcal{Q}$} \\
\hline
\end{tabular}
\caption{Table of notation} %
\label{tab:symbolsForLOICL}
\end{table}
\endgroup

\subsection{Privacy of the Current Location and Location of Interest}\label{sect:ch2-cl}

\begin{wrapfigure}{l}{0.48\textwidth}
  \centering
  \includegraphics[width=\linewidth]{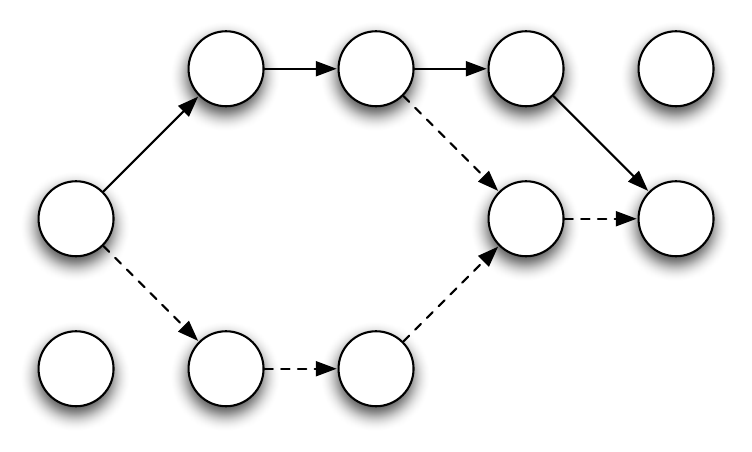}
  \caption{The possible routes from the westernmost central node to the
easternmost central node if fixed shortest-path routing is used (solid lines)
or if random shortest-path routing is used (solid and dashed lines).}\label{fig:routing}
\end{wrapfigure}

How much information a user's queries leak about
the user's current location~(CL) and location of interest~(LOI) depends on how much information the adversary
is able to overhear during
the relaying of the queries.  To determine the CL-privacy and LOI-privacy that result from a set of queries, we
will simulate the user issuing queries while malicious nodes attempt
to ascertain the origin and real destination of those queries.
This section demonstrates how an adversary could use the information captured to narrow down the possible locations where the user could be.

Central to the technique of narrowing down possible locations
is the concept of a \emph{possible route}.  In the network, which we assume to be connected, there will be a routing
algorithm capable of routing messages from any source node to any destination
node.  Consider a route \mbox{$\mathcal{R} = \left(N_{1}, N_{2}, \ldots, N_{l}
\right)$}, which is a sequence of $l$ nodes.  $\mathcal{R}$ is a possible route
from $N_1$ to $N_l$ if it is possible that the routing algorithm used in the 
network could have routed a message from $N_1$ to $N_l$ along the path
\mbox{$N_1, N_2, \ldots, N_l$}.

We will consider two routing algorithms: fixed shortest-path routing
and random shortest-path routing.  Both guarantee that any message from $N_1$ to
$N_l$ will arrive in the fewest possible hops.  In random shortest-path routing,
each node maintains a table indexed by the destination of a message, containing
all possible next hops that the message could take to arrive in the fewest hops. 
When a message arrives at a node $N_i$ destined for node $N_l$, $N_i$ will look
into its table at index $N_l$, and randomly choose one of the entries as the next
hop.  In fixed shortest-path routing, each node stores only a single next-hop
choice for each possible destination.  There is exactly one possible route from
$N_1$ to $N_l$ when fixed shortest-path routing is used, but there can be many
possible routes between $N_1$ and $N_l$ with random shortest-path routing.
An example in Figure~\ref{fig:routing} shows one possible route from the
westernmost central node to the easternmost central node if fixed shortest-path
routing is used, but three possible paths if random shortest-path routing is
used.

Recall that the user employs a mobile device to communicate with a nearby 
node, in order to route queries through the  network to the LOIs.  The user
sends the $k$ queries, denoted \mbox{$\mathcal{Q} = \left\{Q_1, Q_2, \ldots, Q_k\right\}$}, to the closest
 node, $C$.  As simplifying assumptions, we
assume both that the user knows the closest node, and that
the mobile device and $C$ communicate using
low-power communication.  That is, the adversary will only overhear communication
between the mobile device and $C$ if $C$ is compromised.  We also assume that the
adversary knows $k$~---~the implication being that if the adversary overhears
fewer than $k$ of the user's queries, the adversary knows how many queries were
not overheard.  Each $Q_i$ takes a route $\mathcal{R}_i$ through the network,
starting at $C$ and ending at $L_i$.  Denote $\mathcal{R}_i$ as a sequence of
nodes with length $l_i$, \mbox{$\mathcal{R}_i = \left(N_{i,1}, N_{i,2}, \ldots,
N_{i,l_i} \right)$}, where \mbox{$N_{i,1} = C$} and \mbox{$N_{i,l_i} = L_i$}. 
The goals of the adversary are to determine $C$ and to determine the real location of interest,
i.e., to determine $L$.

While our $k$-anonymity scheme can be built over any existing query mechanism, we
assume that query and reply messages are designed to maximize privacy.
Specifically, a query message that is being routed from $C$ to $L_i$ cannot contain
references to $C$, nor can a reply; otherwise, the adversary could easily
determine $C$ (note, though, that even if a reference to $C$ were contained in queries
and replies, this scheme could still be used to preserve the privacy of the user's location
of interest, just not their current location).  Query messages contain four pieces of information in addition to
the query itself: a unique query identifier for $Q_i$; the destination $L_i$
(which may be the real location of interest or a dummy one); the identifier for
the node currently transmitting the query, $N_{i,j}$; and, the next hop in the route,
$N_{i,j+1}$.  When node $N_{i,j+1}$ receives the query, it remembers the previous
node in the route for query $Q_i$, $N_{i,j}$.  Replies to the query message
contain only the query identifier for $Q_i$.  When node $N_{i,j+1}$ receives a
reply to the query, to be routed back to $C$, $N_{i,j+1}$ uses its memory to
identify $N_{i,j}$ as the next hop in the reply path, and sends the reply to
$N_{i,j}$ (without unnecessary information such as the identity of $N_{i,j}$ or
$N_{i,j+1}$).  However, we assume the worst case: the adversary is able to
determine which  node is transmitting a reply message if that reply is
overheard.

If a query message for query $Q_i$ is overheard and the header can be decrypted by the adversary
(e.g., if a compromised node is routing the query, or if a global encryption key is used for all transmissions
in the  network), the adversary learns one of the
LOIs, $L_i$.  Additionally, the adversary learns one hop that the query took
along the route from $C$ to $L_i$: $N_{i,j}$ and $N_{i,j+1}$, for \mbox{$1 \le j
< l_i$}, where $j$ and $l_i$ are unknown to the adversary.  That is, the
adversary learns two consecutive elements in the route, but neither their
position in the route nor the length of the route.  If a reply message for query
$Q_i$ is overheard, the adversary learns only a single node that was involved in
the route: $N_{i,j}$, for \mbox{$1 < j \le l_i$}, again with unknown $j$ and
$l_i$.

The adversary can also construct a list of  nodes that were certainly not
involved in routing $Q_i$.  Because the adversary has complete knowledge of every
message that was routed through compromised nodes in the user's network,
the adversary knows which compromised nodes were not involved in routing $Q_i$. 
An additional consideration for the adversary is that the malicious nodes
could monitor all communication by honest 
nodes in their communication range.  An adversary may conclude that if a
malicious node that is monitoring an honest node within its range did not
hear the honest node produce any traffic regarding query $Q_i$, then the honest node must not have been
involved in routing $Q_i$.  However, the malicious node may not have heard a
message transmitted by the honest node due to interference or a collision.  As such, we
assume that the adversary will restrict the list of nodes that certainly were not
involved with query $Q_i$ to the set of compromised nodes in the network
that did not route $Q_i$.

The key insight for the adversary, having collected information on the routes of
the queries, is that if no possible route from a node $N$ to the known
destination of $Q_i$, $L_i$, is consistent with the known information about
$Q_i$, then $N$ could not have been the origin of the query.  Formally, a route
$\mathcal{R} = \left\{N_1, N_2, \ldots N_l \right\}$ from $N$ to $L_i$ is
\emph{consistent} with that information, assuming the destination $L_i$ of $Q_i$
is known, if: \newline 
\noindent (1) $\mathcal{R}$ is a possible route from $N$ to $L_i$ (i.e., \mbox{$N_1 =
N$}, \mbox{$N_l = L_i$}, and the routing algorithm could have used this route); \newline
\noindent (2) For every query message about $Q_i$ decrypted by the adversary, sent by
$N_{i,j}$ to $N_{i,j+1}$, there is some \mbox{$k < l$} such that \mbox{$N_k =
N_{i,j}$} and \mbox{$N_{k+1} = N_{i,j+1}$}; \newline
\noindent (3) For ever query message about $Q_i$ overheard but not decrypted by the adversary,
sent by $N_{i,j}$ to an unknown destination, there is some
\mbox{$k < l$} such that \mbox{$N_k = N_{i,j}$}; \newline
\noindent (4) For every reply message about $Q_i$ overheard by the adversary, sent by
\mbox{$N_{i,j}$}, there is some \mbox{$k > 1$} such that \mbox{$N_k = N_{i,j}$};
and, \newline
\noindent (5) No node that is known not to have routed $Q_i$ appears in $\mathcal{R}$.

A simple example of the insight provided to the adversary from determining whether or not possible routes
are consistent with known information is illustrated in Figure~\ref{fig:narrow-cls}.  In the illustrated example, there
are two malicious nodes, one of which routed a query to the LOI and another that did not.  Using this information,
the adversary can narrow down the potential set of current locations of the user.

There is a second insight, about queries for which the adversary does not know
the destination.  The adversary may still know some information about such a
$Q_i$ (e.g., non-decrypted query messages,
overheard reply messages, or knowledge about compromised nodes that
did not route $Q_i$).  Denote $\mathcal{D}$ as a set of possible destinations
for the queries with unknown destinations.  Specifically, let $\mathcal{D}$ be the
set of all nodes in $\mathcal{N}$ that are not compromised and are not the
destination of a query with a known destination.  Let $\mathcal{U}$ be the set of
queries for which the destination is unknown to the adversary, where \mbox{$0 \le
\left|\mathcal{U}\right| \le k$}.  If it is not possible to assign a unique
destination from $\mathcal{D}$ to each query in $\mathcal{U}$, in such a way as
to ensure that there is a possible route from node $N$ to the destination
of each query in $\mathcal{U}$ that is consistent with all of the known
information about that query, then $N$ cannot be the origin of the user's
queries.

\begin{wrapfigure}{r}{0.48\textwidth}
  \centering
  \includegraphics[width=\linewidth]{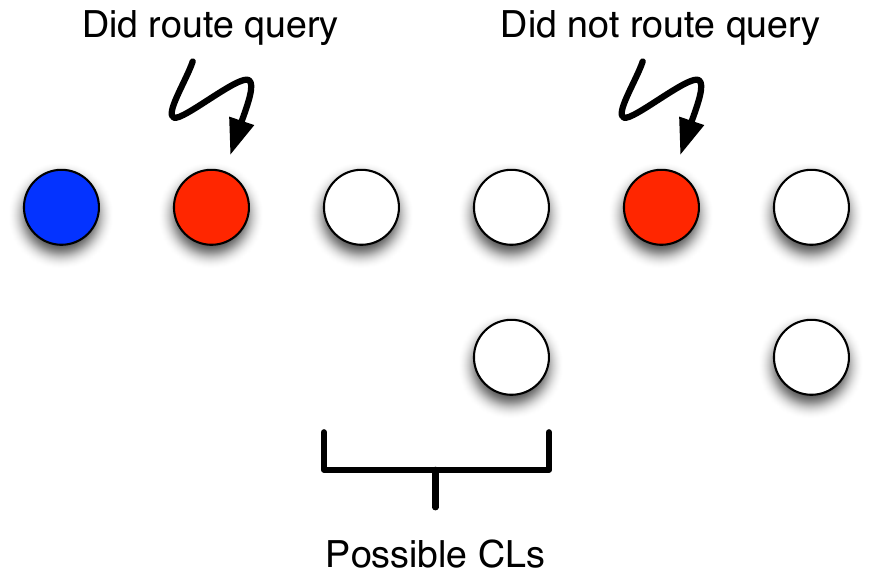}
  \caption{The westernmost (blue) node is one of the queried locations, and knowledge of which
malicious (red) nodes did and did not route the query to the LOI allows the adversary to narrow down the user's potential current locations.}\label{fig:narrow-cls}
\end{wrapfigure}

Using the pseudo-code algorithm presented in Algorithm~\ref{alg:narrow-possible-cls}
and Algorithm~\ref{alg:unknown-naive},
an adversary can narrow down the possible current locations of the user.  Additionally, each assignment
of potential destinations to the queries that is consistent with the information known by the adversary
represents all the possible query destinations for each possible current location.
The pseudocode returns a set $\mathcal{C}$, representing all possible current locations for the user,
and sets \mbox{$\mathcal{L}_\mathcal{N}$} for each \mbox{$N \in \mathcal{C}$}, representing all the possible
locations of interest (i.e., query destinations) for the user if they are indeed at current location $N$.

Algorithm~\ref{alg:narrow-possible-cls} begins, assuming that the user's current location is not compromised,
by considering all honest nodes as potential current locations on line~\ref{algline:honest}.  However, any
honest nodes from which queries consistent with the adversary's known information could not have been sent
are removed as potential current locations on line~\ref{algline:remove}.  Potential query destinations
from each potential current location obviously includes all known destinations, as computed on
line~\ref{algline:known}.  However, the adversary must compute potential destinations for queries with
unknown destinations as well, as done on lines~\ref{algline:startunknown}--\ref{algline:endunknown}.

\begin{algorithm}
\begin{algorithmic}[1]
\Function{Narrow-Possible-CLs-And-LOIs}{}
    \If {the node $C$ with which the user is communicating is compromised}
        \State \Return $\mathcal{C} = \left\{C\right\},\mathcal{L}_C = \left\{L_1,\ldots,L_k\right\}$
    \EndIf
    \State $\mathcal{N}_\mathit{honest} \gets \left\{N \in \mathcal{N}~|~N~\text{is not compromised} \right\}$\label{algline:honest}
    \State $\mathcal{C} \gets \mathcal{N}_\mathit{honest}$
    \State $\mathcal{L}_N \gets \{\}~\forall~N \in \mathcal{C}$
    \ForAll {queries $Q_i$ for which the destination is known}
        \ForAll {$N \in \mathcal{C}$}
            \If {$\Call{Possible-Consistent}{N, Q_i} = \const{False}$}
                \State Remove $N$ from $\mathcal{C}$\label{algline:remove}
            \EndIf
        \EndFor
    \EndFor
    \State $\mathcal{K} \gets \left\{L_i~|~L_i~\text{is known}\right\}$
    \State $\mathcal{U} \gets \left\{Q_i~|~L_i~\text{is unknown} \right\}$ \label{line:startComplex}
    \If {$\left|\mathcal{U}\right| > 0$}\label{algline:startunknown}
    \State $\mathcal{D} \gets \mathcal{N}_\mathit{honest} \setminus \mathcal{K}$
    \ForAll {$N \in \mathcal{C}$}
        \If {\Call{Analyse-Unknown}{$N, \mathcal{L}_N, \mathcal{D}, \mathcal{U}$} = \const{False}}
            \State Remove $N$ from $\mathcal{C}$\label{algline:endunknown}
        \EndIf
    \EndFor
    \EndIf \label{line:endComplex}
    \State $\mathcal{L}_N \gets \mathcal{L}_N \cup \mathcal{K}~\forall~N \in \mathcal{C}$\label{algline:known}
    \State \Return $\mathcal{C},\left\{\mathcal{L}_N~|~N \in \mathcal{C}\right\}$
\EndFunction
\Statex
\Function{Possible-Consistent}{$\mathit{src}, \mathit{Q}$}
    \State $L \gets$ the destination of query $Q$
    \If {$\exists$ a possible route from $\mathit{src}$ to $L$, consistent with everything known about $Q$}
        \State \Return \const{True}
    \Else
        \State \Return \const{False}
    \EndIf
\EndFunction
\end{algorithmic}
\caption{Pseudocode for finding all possible current locations and locations of interest for the user.} \label{alg:narrow-possible-cls}
\end{algorithm}

\begin{algorithm}
\begin{algorithmic}[1]
\Function{Analyse-Unknown}{$N, \mathcal{L}_N, \mathcal{D}, \mathcal{U}$}
    \State $\mathit{foundAssignment} \gets \const{False}$
    \For {$\mathcal{D}^\prime \gets$ each of the $\binom{\left|\mathcal{D}\right|}{\left|\mathcal{U}\right|}$ choices of $\left|\mathcal{U}\right|$ destinations from $\mathcal{D}$}
        \ForAll {$\left|\mathcal{D}^\prime\right|!$ assignments of nodes in $\mathcal{D}^\prime$ as destinations for $\mathcal{U}$}\label{algline:allpossible}
            \If {$\Call{Possible-Consistent}{N, Q} = \const{True}~\forall~Q \in \mathcal{U}$}
                \State $\mathit{foundAssignment} \gets \const{True}$
                \State $\mathcal{L}_N \gets \mathcal{L}_N \cup \mathcal{D}^\prime$\label{algline:addunknown}
            \EndIf
        \EndFor
    \EndFor
    \State \Return $\mathit{foundAssignment}$
\EndFunction
\end{algorithmic}
\caption{A na\"{i}ve implementation of the analysis of queries with unknown destinations} \label{alg:unknown-naive}
\end{algorithm}

The pseudocode presented in Algorithm~\ref{alg:unknown-naive} illustrates the most natural way of
thinking about how the adversary can use information gained about queries with unknown destinations.
The adversary considers all possible assignments of destinations to all queries with unknown locations on
line~\ref{algline:allpossible}, adding any possible destinations to the pool of potential destinations for the
current locations on~\ref{algline:addunknown}.  If there are no possible assignments of destinations to
queries with unknown destinations, the potential current location is removed from the pool of current location
possibilities by a return of $\const{False}$ to Algorithm~\ref{alg:narrow-possible-cls}.
However, the na\"ive implementation of this pseudocode would be highly inefficient.  The following\footnote{All proofs requiring more than a few lines are collected in Appendix \ref{longproofs}.} lemma allows the adversary to implement the algorithm far more efficiently.

\begin{lemma}\label{lemma:set-union}
Given a family of sets of integers, \mbox{$\left\{\mathcal{S}_1, \ldots, \mathcal{S}_m\right\}$}
where \mbox{$\mathcal{S}_i \subseteq \mathbb{Z}$}, we call a set
\mbox{$\mathcal{G}_m = \left\{g_1, \ldots, g_m\right\}$} a \emph{sampling} if
\mbox{$g_i \in \mathcal{S}_i$} for all \mbox{$1 \le i \le m$}, and \mbox{$g_i \neq g_j$} for all \mbox{$i \neq j$}.
Let \mbox{$\mathcal{G}_m^\prime = \left\{\mathcal{G}_m~|~\mathcal{G}_m~\text{is a sampling}\right\}$}, i.e.,
$\mathcal{G}_m^\prime$ is a set of sets, where each
\mbox{$\mathcal{G}_m \in \mathcal{G}_m^\prime$} is a set of $m$ unique integers, with each integer coming from
a different one of the $\mathcal{S}_i$, \mbox{$1 \le i \le m$}.
Then either \mbox{$\mathcal{G}_m^\prime = \emptyset$} or
\begin{equation*}
\bigcup_{\mathcal{G}_m \in \mathcal{G}_m^\prime}\mathcal{G}_m = \bigcup_{1 \le i \le m}\mathcal{S}_i\;.
\end{equation*}
\end{lemma}

Given a potential current location \mbox{$N \in \mathcal{C}$} and queries with unknown destinations $\mathcal{U}$,
we can compute potential destinations \mbox{$\mathcal{S}_1, \ldots, \mathcal{S}_{|\mathcal{U}|}$} for each of
the queries.
We can then test if it is possible to choose a unique destination for each query, i.e., choose
\mbox{$s_i \in \mathcal{S}_i$}
with \mbox{$s_i \neq s_j$} for all \mbox{$i \neq j$}, where \mbox{$1 \le i,j \le |\mathcal{U}|$}.
If it is possible to choose \emph{even a single set} of unique destinations,
then Lemma~\ref{lemma:set-union} tells us that the set of all
potential destinations for all the queries
with unknown destinations is \mbox{$\cup_{1 \le i \le |\mathcal{U}|}\mathcal{S}_i$}, i.e., no potential
destinations can be excluded.
Testing whether it is possible to choose a unique destination from each $\mathcal{S}_i$ can be accomplished
by creating a graph \mbox{$G = \langle(X, Y), E\rangle$}, where $X$ contains a vertex for each of the
\mbox{$|\mathcal{U}|$} sets \mbox{$\mathcal{S}_i$}, $Y$ contains a vertex for each potential destination in
\mbox{$\cup_{1 \le i \le |\mathcal{U}|}\mathcal{S}_i$}, and $E$ contains an edge between the
vertex in $X$ representing $\mathcal{S}_i$ and the
vertex in $Y$ representing $s$ if and only if \mbox{$s \in \mathcal{S}_i$}.
It is possible to choose a unique destination from each $\mathcal{S}_i$ if and only if there is a bipartite matching
on $G$ of size $|\mathcal{U}|$, and testing for a bipartite matching is a
well-studied problem~\mbox{\cite[pp.~664--669]{cormen:algorithms}}.
A revised, efficient algorithm for analyzing queries with unknown destinations is illustrated as
Algorithm~\ref{alg:unknown-efficient}.

\begin{algorithm}
\begin{algorithmic}[1]
\Function{Analyse-Unknown}{$N, \mathcal{L}_N, \mathcal{D}, \mathcal{U}$}
    \State $\mathcal{S}_i \gets \{\}~~\forall~1 \le i \le |\mathcal{U}|$
    \ForAll {$U_i \in \mathcal{U}$}
        \ForAll {$D \in \mathcal{D}$}
            \If {\Call{Possible-Consistent-Assume}{$N, D, U_i$}}
                \State Add $D$ to $\mathcal{S}_i$
            \EndIf
        \EndFor
    \EndFor
    \State $X \gets$ a vertex for each $\mathcal{S}_i, 1 \le i \le |\mathcal{U}|$
    \State $Y \gets$ a vertex for each destination in $\cup_{1 \le i \le |\mathcal{U}|}\mathcal{S}_i$
    \State $E \gets$ edges between vertex representing $\mathcal{S}_i$ and vertex representing $D$ iff $D \in \mathcal{S}_i$
    \If{there exists a bipartite matching of size $|\mathcal{U}|$ on $G=\langle(X,Y),E\rangle$}
        \State $\mathcal{L}_N \gets \mathcal{L}_N \cup \mathcal{S}_1 \cup \ldots \cup \mathcal{S}_{|\mathcal{U}|}$
        \State \Return \const{True}
    \Else
        \State \Return \const{False}
    \EndIf
\EndFunction
\Statex
\Function{Possible-Consistent-Assume}{$\mathit{src}, L, \mathit{Q}$}
    \State Assume, just for now, that the destination of $Q$ is $L$
    \If {$\exists$ a possible route from $\mathit{src}$ to $L$, consistent with all known/assumed about $Q$}
        \State \Return \const{True}
    \Else
        \State \Return \const{False}
    \EndIf
\EndFunction
\end{algorithmic} 
\caption{An efficient implementation of the analysis of queries with unknown destinations}\label{alg:unknown-efficient}
\end{algorithm}

Given the adversary's ability to narrow down the user's possible current locations to a set $\mathcal{C}$
and possible locations of interest to a set \mbox{$\mathcal{L}_N$} for each \mbox{$N \in \mathcal{C}$},
a simple yet expressive way to measure the security of a user's location of interest and current location
is the probability that the adversary will guess the LOI and CL incorrectly, respectively.
As a simple example, if the adversary knows the value of $k$
(a necessary assumption to make, as it is unsafe to assume otherwise), and if the
adversary overhears all $k$ queries in a uniform, featureless environment
(i.e., all $n$ potential query destinations in the network are equally probable destinations),
then the user's resulting LOI-privacy will be $\frac{k-1}{k}$, because the adversary can
merely guess at random that one of the $k$ destinations is the correct one.

More formally, we begin by assuming the adversary knows~---~from a global, aggregate perspective~---~the
probability that a given node in the network is a user's current location, and the probability that each node is
a user's location of interest given a known current location.  We denote these probabilities as
\probC{N} and \probL{M}{N} respectively.  For now, we assume that
\mbox{$\probC{N} = \probL{M}{N} = \frac{1}{n}$} for all nodes $N$ and $M$ in the network.
We will revisit this assumption in Section~\ref{sect:ch2-mobility}.

First, the adversary corrects the probability of each node being the current location, based on the set of possible
current locations.  That is, the adversary computes
\begin{equation*}
\probpC{N} =\begin{cases}
\displaystyle\frac{\probC{N}}{\sum_{P \in \mathcal{C}}{\probC{P}}} & \mbox{if $N \in \mathcal{C}$}\\
0 & \mbox{if $N \not\in \mathcal{C}$}\end{cases}
\end{equation*}
for each \mbox{$N \in \mathcal{C}$}.  From the computation of this corrected probability, we naturally
get the formal definition of CL-privacy, which is the probability of the adversary incorrectly guessing the
user's current location.  Define
\begin{equation*}
C^* = \argmax_{N \in \mathcal{C}} \left\{ \probpC{N} \right\} \;,
\end{equation*}
i.e., \mbox{$C^*$} is the adversary's best guess as to the user's current location.
The formal definition of CL-privacy can then be stated as
\begin{equation*}
\proc{M}_\mathit{CL} = 1 - \probpC{C^*} \;.
\end{equation*}

The user's LOI-privacy is computed in a similar fashion: the probability that the adversary incorrectly guesses
the user's actual location of interest.
First, note that the user's actual location of interest may be a destination of one of the $k$ queries where the
destination is known to the adversary (e.g., if a compromised node routed the query), or it may be the
destination of a query where the destination is unknown to the adversary.  Based on the value of
\probL{M}{N}, an adversary may actually guess that the user's real location of interest is a
location that is not known to be a query destination, if the adversary does not know the destination of all $k$ queries.
To account for this possibility, we first need to apply a ``corrective weight''
to the probability that a given
location \mbox{$M \in \mathcal{N}$} is the user's real LOI, assuming that another location
\mbox{$N \in \mathcal{N}$} is their current location.
Let $\mathcal{K}$ be the set of all query destinations that are known to
the adversary (i.e., \mbox{$\mathcal{K} \subseteq \mathcal{L}$}), then denote
\begin{equation*}
\chi(N, M) =\begin{cases}
1 & \mbox{if $M \in \mathcal{K}$} \\
0 & \mbox{if $M \not\in \mathcal{K}$ and $|\mathcal{L}_N| = |\mathcal{K}|$} \\
\displaystyle\frac{k - |\mathcal{K}|}{|\mathcal{L}_N| - |\mathcal{K}|} & \mbox{otherwise.}\end{cases}
\end{equation*}

Using $\chi$ as a corrective weight, we can compute the weighted probability that $M$ is the user's location
of interest, given that $N$ is their current location and their location of interest is in \mbox{$\mathcal{L}_N$}, as
\begin{equation*}
\phi(N, M) =\begin{cases}
\displaystyle\frac{\probL{M}{N} \cdot \chi(N, M)}
{\sum_{P \in \mathcal{L}_N}{\probL{P}{N} \cdot \chi(N, P)}}
& \mbox{if $M \in \mathcal{L}_N$} \\
0 & \mbox{$if M \not\in \mathcal{L}_N$\;,}\end{cases}
\end{equation*}
for each \mbox{$N \in \mathcal{C}$} and \mbox{$M \in \mathcal{L}_N$}.
We subsequently define \mbox{$\psi(N, M)$} as the probability that both
$N$ is the current location and $M$ is the real location of interest.  Namely: $\psi(N, M) = \probpC{N} \cdot \phi(N,M) \;.$

We assume that the goal for the adversary is to choose the most likely real LOI, independent of the adversary's
guess about the user's current location (we return to this assumption shortly).
In this case, the adversary would compute
the probability that a node $M$ is the user's real location of interest over all possible current locations as
\begin{equation*}
\omega(M) = \sum_{N \in \mathcal{C}}{\psi\left( N, M \right)} \;,
\end{equation*}
for all possible LOIs $M$ (i.e., for any \mbox{$M \in \cup_{N \in \mathcal{C}}{\mathcal{L}_N}$}).
The adversary then guesses the real LOI as \mbox{$L^* = \argmax_M \omega(M)$}.  As such,
the user's resulting LOI-privacy is
\begin{equation*}
\proc{M}_\mathit{LOI} = 1 - \omega\left(L^*\right) \;.
\end{equation*}
Given the above definitions, it is worth stating two consequences, posed as the following lemmas. 

\begin{lemma}\label{lemma:equal}
If the adversary knows that any user's real location of interest is in fact their current location, i.e.,
\begin{equation*}
\probL{N}{M} =\begin{cases}
0 & \mbox{if $N \neq M$} \\
1 & \mbox{if $N = M$\;,}\end{cases}
\end{equation*}
then \mbox{$\proc{M}_\mathit{CL} = \proc{M}_\mathit{LOI}$}.

\end{lemma}
\begin{lemma}\label{lemma:compromised}
If the user's current location is compromised, then
\begin{equation*}
\proc{M}_\mathit{CL} = 0
\end{equation*}
and
\begin{equation*}
\proc{M}_\mathit{LOI} =
1 - \max_{M \in \mathcal{L}}{\probGlobal\left(L = M~|~C = C^*~\text{and}~L \in \mathcal{L}\right)} \;.
\end{equation*}
Consider a uniform, featureless environment~---~that is, an environment
in which, from any current location, the probability of each location
being any user's real LOI is $\frac{1}{n}$.
If the user's current location is compromised
in a uniform, featureless environment, then
\begin{equation*}
\proc{M}_\mathit{LOI} = (k-1)/k\;. %
\end{equation*}
\end{lemma}

Previously, we assumed that the adversary would guess the user's current location and real location of interest
independently of one another (i.e., a best attempt to get at least one, if not both, of those two locations correct).
In some contexts, however (e.g., intercepting a user who is moving from their current location to a location of
interest), it might be in the interest of the adversary to choose the most likely
\mbox{$\langle C^*, L^* \rangle$} pair.  That is, the adversary wants to maximize their chances of guessing
\emph{both} the user's CL and LOI simultaneously correctly.  Such an adversary would choose $C^*$ and
$L^*$ as $\argmax_{\langle C^*, L^* \rangle}{\psi\left(C^*, L^*\right)} \;,$
and the computation of the user's CL- and LOI-privacy would remain the same as $\proc{M}_\mathit{CL} = 1 - \probpC{C^*}$ and $\proc{M}_\mathit{LOI} = 1 - \omega\left(L^*\right)$ respectively.
For the purposes of this paper, however, we restrict our investigation to an adversary who guesses
the user's current location and location of interest independently,
as such an adversary is stronger.

\subsection{Privacy of the Area of Interest}\label{sect:ch2-aoi}

If it is equally probable that any of the $k$ queried nodes is the real LOI (i.e., in the featureless environment),
then the adversary
cannot learn which of the $k$ nodes queried is the real LOI.
Given this assumption,
what is actually meant when we discuss protecting the
privacy of the user's LOI?  Consider the two scenarios in Figure~\ref{fig:aoi}. 
In both cases, the user's real LOI is the starred node in the northeast, and
\mbox{$k-1=3$} dummy LOIs are chosen to disguise the real LOI.  However, in
Figure~\ref{fig:aoi-poor}, all four LOIs are clustered in the east.  While an
adversary who overhears these queries would not know which node is of interest to
the user, it would be obvious that the user is interested in the eastern region
of the network.  In Figure~\ref{fig:aoi-good}, the four LOIs are dispersed
throughout the network, obfuscating the area of the network in which the
user is interested.  We need to define a measure of how dispersed the $k$ LOI
choices are.  That is, we want to measure how well-protected the privacy of the
user's area of interest~(AOI) is.

\begin{figure*} 
\centering 
\begin{subfigure}[b]{0.48\textwidth}
\centering
\includegraphics[width=\textwidth]{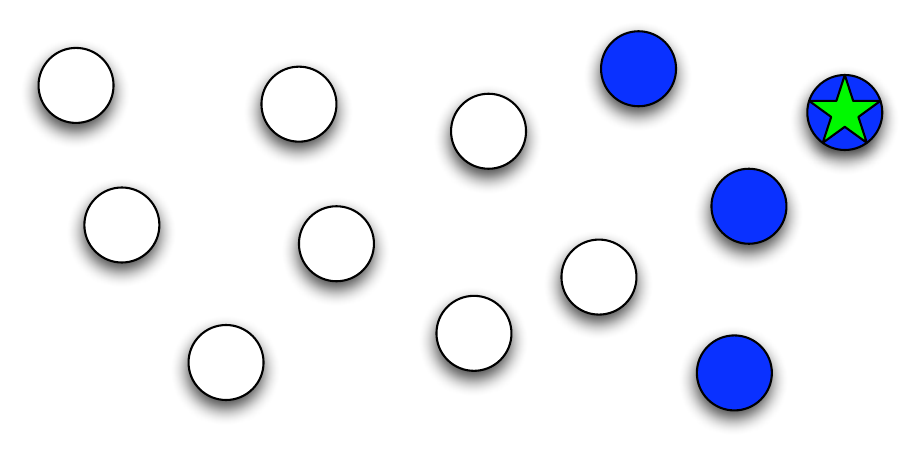}
\caption{A poor choice.}
\label{fig:aoi-poor}
\end{subfigure}
\hfill
\begin{subfigure}[b]{0.48\textwidth}
\centering
\includegraphics[width=\textwidth]{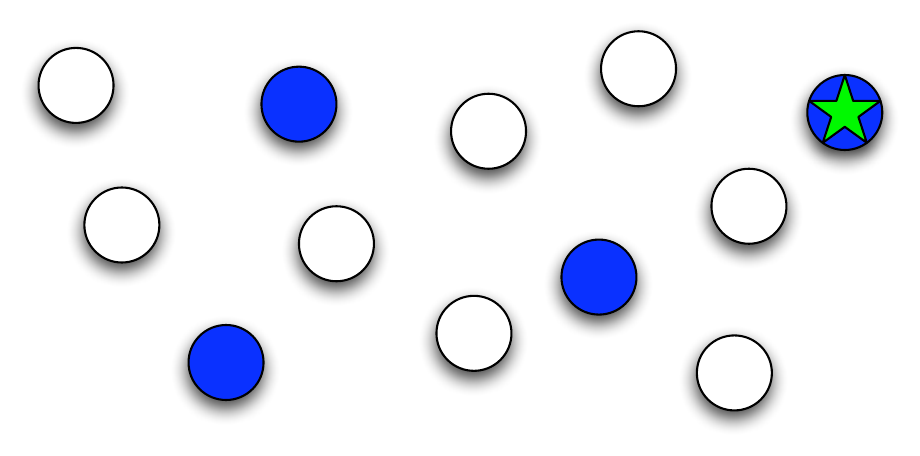}
\caption{A good choice.}
\label{fig:aoi-good}
\end{subfigure}
\caption{A comparison of two choices for \mbox{$k-1=3$} dummy LOIs (dark nodes) given one fixed real LOI (starred node).} 
\label{fig:aoi} 
\end{figure*}

To measure how well the set $\mathcal{L}$ of $k$ LOIs preserves the user's
AOI-privacy, we define a function \mbox{$\sigma\left(\mathcal{L},
\mathcal{N}\right)$} to measure how dispersed the LOIs are over the network
$\mathcal{N}$.  To allow for comparisons of different methods of choosing the
dummy LOIs over networks with different topologies, we normalize the score
returned by $\sigma$.  Let \mbox{$\sigma_\mathit{min}(k, \mathcal{N})$} and
\mbox{$\sigma_\mathit{max}(k,\mathcal{N})$} be the minimal and maximal values
returned by $\sigma$, over all \mbox{$\binom{n}{k}$} possible sets of $k$ LOIs. 
The normalized measure of AOI-privacy is defined as
\begin{equation*}
\proc{M}_\mathit{AOI}\left(\mathcal{L}, \mathcal{N}\right) =\begin{cases}
     1 & \mbox{if $\sigma_\mathit{min}\left(k, \mathcal{N}\right) = \sigma_\mathit{max}\left(k, \mathcal{N}\right)$} \\
     \dfrac{\sigma\left(\mathcal{L}, \mathcal{N}\right) - \sigma_\mathit{min}\left(k, \mathcal{N}\right)}{\sigma_\mathit{max}\left(k, \mathcal{N}\right) - \sigma_\mathit{min}\left(k, \mathcal{N}\right)} & \mbox{otherwise.}
\end{cases}
\end{equation*}
The function $\sigma$ must have the property that it returns large values for
sets $\mathcal{L}$ with minimal clustering of the LOIs, and small values
otherwise.  Ideally, it is also easy to compute on the user's low-powered, mobile
device prior to issuing a query.  Next, we discuss some alternatives for
$\sigma$.

\subsubsection{Variance-Based $\sigma$.}

One straightforward approach is to compute a variance-like quantity for the
positions of the nodes in $\mathcal{L}$, measuring the squared distance between
LOIs:
\begin{equation*}
\sigma\left(\mathcal{L}, \mathcal{N}\right) = \sum_{i=1}^{k-1} \sum_{j=i+1}^{k} \proc{D}\left(L_i, L_j\right)^2 \;,
\end{equation*}
where \proc{D} is the Euclidean distance between two nodes.  However, this
definition of $\sigma$ does not penalize clustering properly.  Consider the
examples in Figure~\ref{fig:variance}.  The choice of dummy LOIs in
Figure~\ref{fig:variance-high} has two of the LOIs clustered together.  In
Figure~\ref{fig:variance-low} the three LOIs are dispersed evenly.  However, the
$\sigma$ value for the first choice is higher.
This result is similar to how the variance for the data set \mbox{$\left\{1,2,5\right\}$} is higher than
the variance for the data set \mbox{$\left\{1,3,5\right\}$}.

\begin{figure*} 
\centering 
\begin{subfigure}[b]{0.48\textwidth}
\centering
\includegraphics[width=\textwidth]{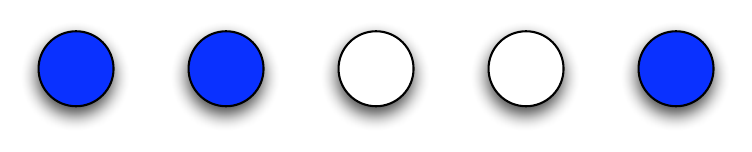}
\caption{A choice of LOIs that results in two clusters.}
\label{fig:variance-high}
\end{subfigure}
\hfill
\begin{subfigure}[b]{0.48\textwidth}
\centering
\includegraphics[width=\textwidth]{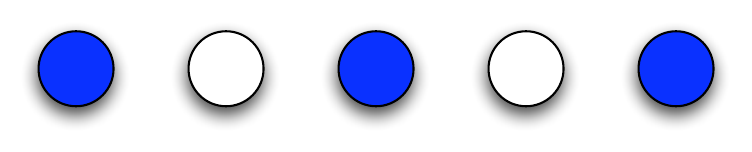}
\caption{A choice of LOIs that disperses the LOIs across the entire network.}
\label{fig:variance-low}
\end{subfigure}
\caption{A comparison of two choices for \mbox{$k=3$} LOIs in a five-node network.} 
\label{fig:variance} 
\end{figure*}

\subsubsection{Union of Circles $\sigma$.}\label{sect:circles}

Another approach is to determine how much overlap exists among the regions around
each LOI.  Specifically, one can draw a circle of radius $r$ around each LOI and
let $\sigma$ be the area of the union of the $k$ circles.  The more clustered the
LOIs, the more overlap there would be among the circles, resulting in a smaller
area.  An example is presented in Figure~\ref{fig:monte}.  With $k$ clustered
LOIs, shown in Figure~\ref{fig:monte-bad}, there is significant overlap among the
circles, unlike with the non-clustered choice shown in
Figure~\ref{fig:monte-good}.

\begin{figure*} 
\centering 
\begin{subfigure}[b]{0.48\textwidth}
\centering
\includegraphics[width=\textwidth]{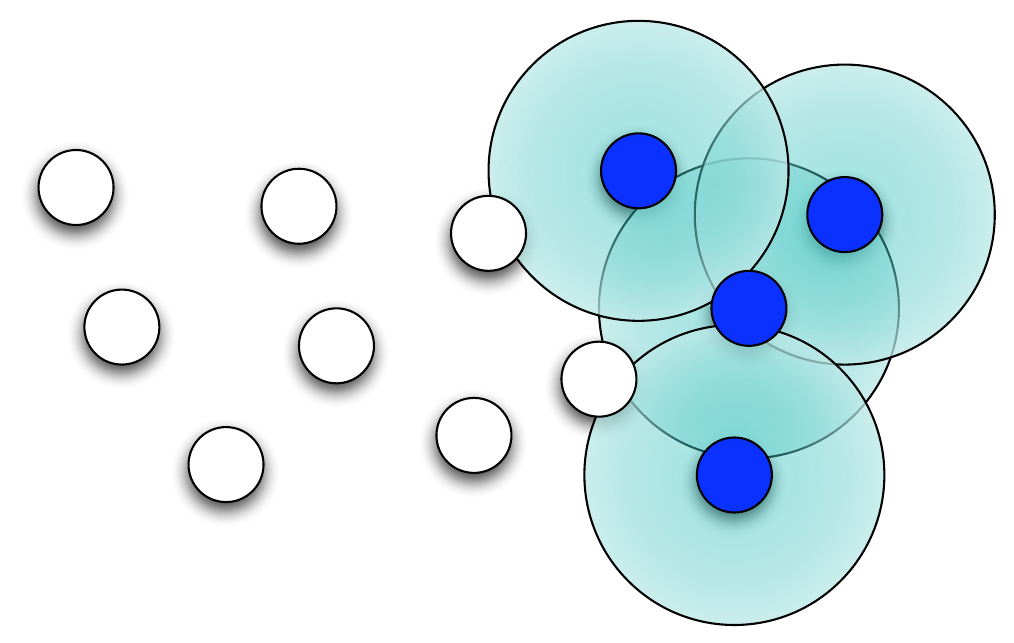}
\caption{A choice of LOIs with significant overlap among the circles.}
\label{fig:monte-bad}
\end{subfigure}
\hfill
\begin{subfigure}[b]{0.48\textwidth}
\centering
\includegraphics[width=\textwidth]{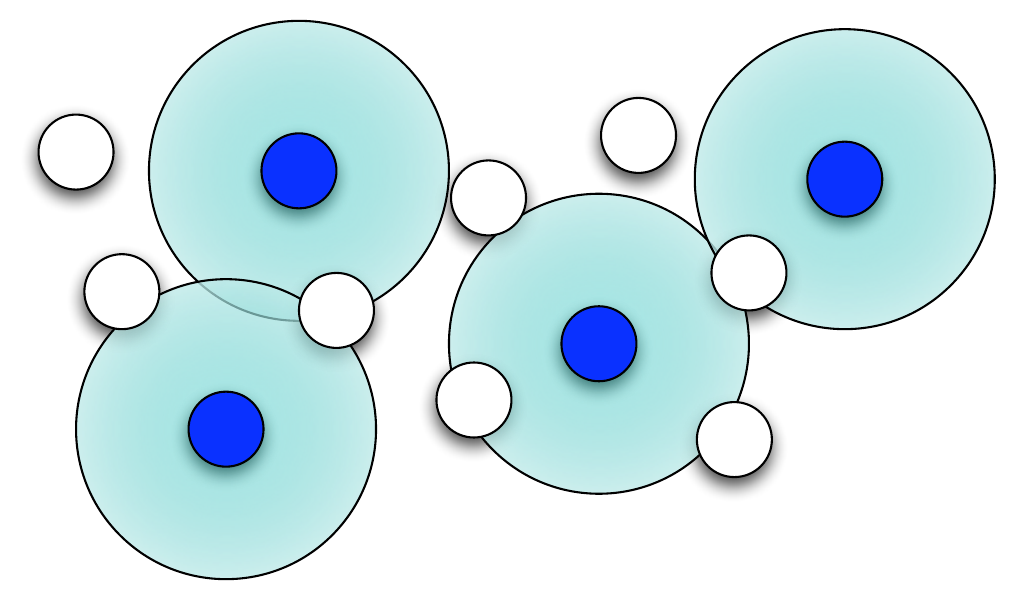}
\caption{A choice of LOIs with minimal overlap among the circles.}
\label{fig:monte-good}
\end{subfigure}
\caption{A comparison of two choices for \mbox{$k=4$} LOIs in a network, each surrounded by a circle with a fixed, arbitrary radius.} 
\label{fig:monte} 
\end{figure*}

One way to choose $r$ is to define $d$ as the maximum distance between any two
nodes in $\mathcal{N}$, then set \mbox{$r = \frac{d}{2}$}.  This definition
ensures that there is no overlap between the regions of the two nodes that
are farthest apart, but the regions surrounding any two nodes that are
closer together will overlap.  Alternately, $r$ could be chosen as the maximal size
such that $k$ nodes could be chosen from $\mathcal{N}$ with no overlap among the
$k$ associated circles.

While this $\sigma$ properly penalizes clustering, it is unclear what an ideal value
for $r$ would be.  Ultimately, we do not investigate this issue further, because
this $\sigma$ is unnecessarily computationally expensive, compared to
the other $\sigma$ alternatives mentioned subsequently.  Note, however, that
there are known methods for computing the area of the union of circles,
including an \mbox{$\mathcal{O}(k^2)$} deterministic algorithm and an
\mbox{$\mathcal{O}(k)$} Monte Carlo algorithm~\cite{spirakis:circles}.

\subsubsection{Minimal Pairwise Distance $\sigma$.}\label{sect:minimalpairwise}

A simple computation that penalizes two LOIs being close together is to measure the minimal
distance between any two LOIs.  Formally, let
\begin{equation*}
\sigma\left(\mathcal{L}, \mathcal{N}\right) = \min_{j \neq i}\left\{ \proc{D}\left(L_i, L_j\right)\right\} \;,
\end{equation*}
where \proc{D} is the Euclidean distance between two nodes.  The closer the two closest $L_i$ become,
the lower the value of $\sigma$.  Unfortunately, while this approach properly punishes having two queried nodes
close together, it does not distinguish between a mediocre situation (e.g., having two of the $L_i$ close together) and
a bad situation (e.g., having all of the $L_i$ close together).  Consider the situation in Figure~\ref{fig:pairwise}.
In Figure~\ref{fig:pairwise-bad}, the situation is not ideal: two of the LOIs are close together; that said, it is not a terrible
physical distribution of the queried location.  Figure~\ref{fig:pairwise-verybad}, on the other hand, illustrates a far worse
situation, in which all of the LOIs are clustered.  In both cases, however, the value of
\mbox{$\sigma\left(\mathcal{L}, \mathcal{N}\right)$} will be the same: the distance between the two closest nodes
(i.e., the two nodes in the upper right).

\begin{figure*} 
\centering 
\begin{subfigure}[b]{0.48\textwidth}
\centering
\includegraphics[width=\textwidth]{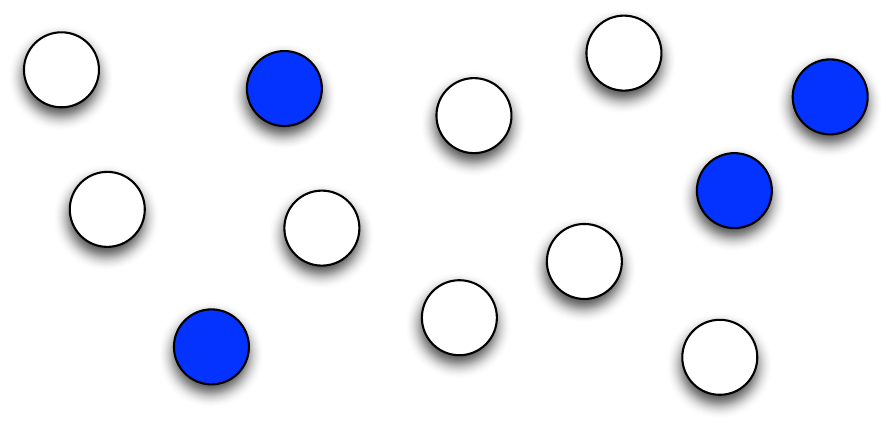}
\caption{A mediocre choice with two LOIs close together.}
\label{fig:pairwise-bad}
\end{subfigure}
\hfill
\begin{subfigure}[b]{0.48\textwidth}
\centering
\includegraphics[width=\textwidth]{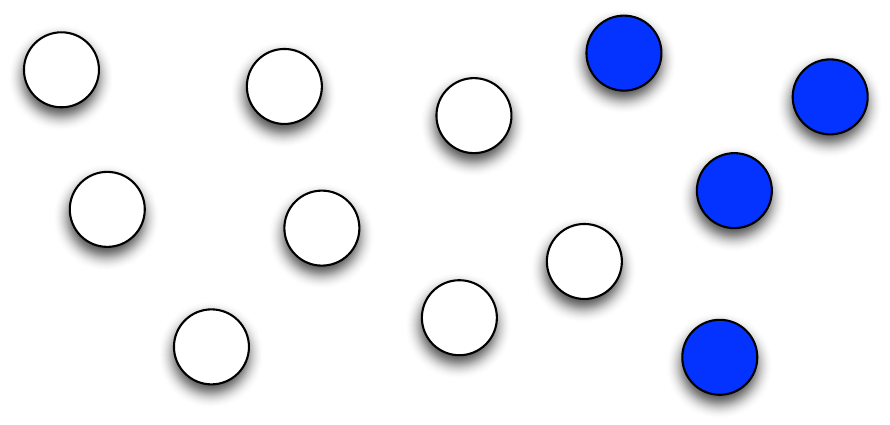}
\caption{A bad choice with all of the LOIs close together.}
\label{fig:pairwise-verybad}
\end{subfigure}
\caption{A comparison of two choices for \mbox{$k=4$} LOIs in a network, both having at least two LOIs close together.} 
\label{fig:pairwise} 
\end{figure*}

\subsubsection{Sum of Minima $\sigma$.}\label{sect:sumofmin}

Measuring the sum of minimum distances from each LOI to any other LOI penalizes
clustering properly and is easy to compute.  Formally, let
\begin{equation*}
\sigma\left(\mathcal{L}, \mathcal{N}\right) = \sum_{i=1}^{k} \min_{j \neq i}\left\{ \proc{D}\left(L_i, L_j\right) \right\} \;,
\end{equation*}
where \proc{D} is again the Euclidean distance between two nodes.  That is, the farther
away each of the $L_i$ are from each other, the higher the value of $\sigma$. 
Returning to the five-node examples in Figure~\ref{fig:variance}, this $\sigma$
returns a higher value in the scenario where the LOIs are non-clustered; and, it
can be computed quickly, given the location of every node in $\mathcal{L}$.
This $\sigma$ function is essentially a more expressive version of the minimal-pairwise-distance
version of $\sigma$, in that it not only properly punishes clustering, but also properly distinguishes
between mediocre and bad distributions of queried locations.

\bigskip \noindent
Considering these four possible $\sigma$ functions, we use the sum of minima
function for most of this paper (it will be noted explicitly when one of the
other $\sigma$ functions is used).  We chose this $\sigma$ function since we assume
that the user, prior to issuing any set of $k$ queries, will want to precompute
the AOI-privacy that will result from the queries using a mobile device with
limited processing power.  If processing power were not a concern, the union of
circles metric could be used instead (though how to best choose the radius of
the circles used in the computation remains an open question).

\section{Choosing the Dummy LOIs}\label{sect:choose-fake}

Given the metrics necessary to measure how well a set of $k$ queries preserves
the privacy of the user's area of interest~(AOI) and current location~(CL), we
now investigate how the user should choose the \mbox{$k-1$} dummy locations of
interest~(LOIs) to query, given one real LOI.  Regardless of how the dummy LOIs
are chosen, recall our assumption that, from the point of view of the adversary,
it is equally probable that any of the $k$ queried nodes is the real LOI.  Any
implementation of $k$-anonymity for network queries must take into account
real-life limitations.  For example, if troops using a military 
network queried one node at a location of strategic importance and \mbox{$k-1$}
nodes at strategically irrelevant positions, the adversary could guess the real
LOI with high probability.

Another situation that could leak the real LOI is if the user issues multiple
queries to the real LOI.  Any algorithm used to choose the \mbox{$k-1$} dummy LOIs
based on the real LOI is required to choose the same \mbox{$k-1$} dummy LOIs each
time the user issues a real query.  Otherwise, if the user first queries the $k$
nodes in $\mathcal{L}_1$ then later issues queries to the $k$ nodes in
\mbox{$\mathcal{L}_2 \neq \mathcal{L}_1$}, and if the adversary can correctly
guess that the user was issuing repeat queries to the same node (it is overly
optimistic to assume otherwise), the adversary would learn that the real LOI is
in \mbox{$\mathcal{L}_1 \cap \mathcal{L}_2$}.

Let \proc{F} be the choice function that takes the real LOI $L$ and returns a set
$\mathcal{L}$ of $k$ LOIs to query, with \mbox{$L \in \mathcal{L}$}.  If \proc{F}
is deterministic, repeat queries are of no concern, since
\mbox{$\proc{F}\left(L\right) = \mathcal{L}$} for each call to
\mbox{$\proc{F}\left(L\right)$}.  If \proc{F} is a random function, taking $L$
and a random seed $s_L$ as parameters and using a
cryptographically secure
pseudo-random number generator~\cite{blum:prng} to choose the dummy LOIs,
the user should encapsulate \proc{F} in a
deterministic function.  For example, the user could call
\mbox{$\proc{F}^\prime\left(L\right) = \proc{F}\left(L, s_L\right)$}, where
\mbox{$s_L = \proc{PRF}\left(s, L\right)$} is computed using a stored secure
seed (i.e., a key), $s$, and pdeudorandom function $\proc{PRF}$.  The pseudorandom
function could be as simple as a cryptographic hash function, e.g., SHA-2~\cite{nist:shs},
on the concatenation of $s$ and $L$, i.e.,
\mbox{$\proc{PRF}\left(s, L\right) = \proc{SHA-2-256}\left(s \circ L\right)$}, with $\circ$ representing
concatenation.  Alternately, a more robust approach could use a key-derivation function~(KDF)
as the pseudorandom function~\cite{nist:kdf}.

It is not only multiple queries to the real LOI that pose a problem; even a
single set of $k$ queries could leak the real LOI.  Knowledge of how \proc{F}
works could be sufficient for the adversary to determine the real LOI, given the
set $\mathcal{L}$ of all $k$ nodes queried.  Consider a deterministic \proc{F}
that takes the real LOI $L$ and returns a list $\mathcal{L}$ that contains $L$
and the \mbox{$k-1$} dummy LOIs, such that \mbox{$\sigma(\mathcal{L},
\mathcal{N})$} is maximized over all possible choices of \mbox{$k-1$} dummy LOIs. 
This function, which maximizes the user's AOI-privacy, may leak the user's real
LOI.  Consider the example in Figure~\ref{fig:leak}.  If the \mbox{$k=3$} nodes
illustrated in Figure~\ref{fig:leak1} are queried, the adversary would know that
the interior LOI is the real LOI.  Were the westernmost node the real LOI, the
central node and the easternmost node would have been chosen by \proc{F} as the
dummy LOIs, as illustrated in Figure~\ref{fig:leak2}, to maximize $\sigma$. 
Similarly, the easternmost node could not be the real LOI.

\begin{figure*} 
\centering 
\begin{subfigure}[b]{0.48\textwidth}
\centering
\includegraphics[width=\textwidth]{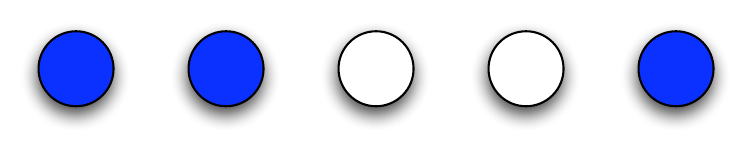}
\caption{The \mbox{$k=3$} nodes that were actually queried by the user.}
\label{fig:leak1}
\end{subfigure}
\hfill
\begin{subfigure}[b]{0.48\textwidth}
\centering
\includegraphics[width=\textwidth]{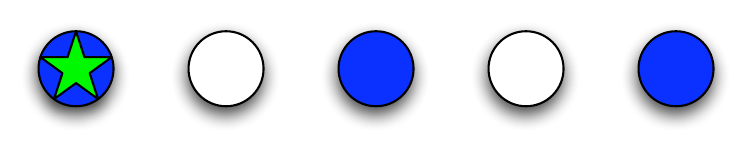}
\caption{The nodes that would have been chosen as $\mathcal{L}$ were the westernmost node the real LOI.}
\label{fig:leak2}
\end{subfigure}
\caption{An example of the $\sigma$-maximizing choice function leaking the real LOI.} 
\label{fig:leak} 
\end{figure*}

For \proc{F} not to leak information in this fashion, \proc{F} must generate a
set \mbox{$\mathcal{L} \subseteq \mathcal{N}$} that is closed under \proc{F}. 
For a deterministic \proc{F}, this closure requirement means: if
\mbox{$\proc{F}\left(L\right) = \mathcal{L}$} for some \mbox{$L \in
\mathcal{N}$}, then \mbox{$\proc{F}\left(L_i\right) = \mathcal{L}$} for all
\mbox{$L_i \in \mathcal{L}$}.
The closure requirement for a non-deterministic \proc{F} is more complex.  We
assume that the random seed $s$ passed to \proc{F} is secret; because the seed
must remain secret, it must be sufficiently long (e.g., 256 bits) to prevent
exhaustive search of the seed space by the adversary.  As such, the adversary
cannot compute \mbox{$\proc{F}\left(L_i, s\right)$} for every \mbox{$L_i \in
\mathcal{L}$}.  What the adversary can do is estimate probabilities.  Define
\mbox{$\proc{P}\left(\proc{F}, \mathcal{L}, L_i\right)$} as the proportion of all
seeds in the seed space for which \proc{F} will generate $\mathcal{L}$, given
$L_i$ and a seed as arguments.  Given $\mathcal{L}$, the adversary can estimate
\mbox{$\proc{P}\left(\proc{F}, \mathcal{L}, L_i\right)$} for each \mbox{$L_i \in
\mathcal{L}$} by running \mbox{$\proc{F}\left(L_i, s_\mathit{test}\right)$} with
a large number of different test seed values.  The adversary can then predict
that the largest estimate corresponds to the most likely real LOI.  That is, if
\mbox{$L_\textit{max} = \argmax_{L_i \in
\mathcal{L}}\left\{\proc{P}\left(\proc{F}, \mathcal{L}, L_i\right)\right\}$},
then $L_\textit{max}$ has the highest probability of having generated
$\mathcal{L}$, and hence is the most likely real LOI.  Based on this attack, the
closure property necessary for a non-deterministic \proc{F} is: if
\mbox{$\proc{F}\left(L, s\right) = \mathcal{L}$} for some \mbox{$L \in
\mathcal{N}$} and seed $s$, then \mbox{$\proc{P}\left(\proc{F}, \mathcal{L},
L\right) = \proc{P}\left(\proc{F}, \mathcal{L}, L_i\right)$} for all \mbox{$L_i \in \mathcal{L}$}.

The simplest method of fulfilling the closure requirement is to generate an unbiased random
partitioning of the network.  When the user wishes to query a real LOI, the dummy LOIs also queried
are the other nodes in the real LOI's partition.  Note that the size of the partitions, while obviously
based on the LOI-privacy requirements of the user, may not be able to be uniform.  For example,
a network of 50 nodes cannot be divided into partitions each of size \mbox{$k = 3$}.

\begin{lemma}\label{lemma:partition-size}
Any network with $n \ge 3$ nodes can be partitioned such that each partition contains $3$--$5$ nodes.
\begin{proof}
By induction.  For \mbox{$n \in [3,5]$}, use a single partition of size $n$.
Assuming you can partition a network of size \mbox{$n = 3j + k$}, where \mbox{$j \ge 0$},
\mbox{$k \in [3, 5]$}, you can partition a network of size \mbox{$n+3 = 3(j+1) + k$} by using an
additional partition of size 3.
\end{proof}
\end{lemma}

Because of Lemma~\ref{lemma:partition-size}, we will use partitions of sizes 3-5 whenever the
need for partitioning a network arises throughout this paper.  We further investigate the problem
of generating partitions and partitionings on networks in the next section.

\section{The Partition and Partitioning Problems}\label{sect:partition-partitioning}

The goal of maximizing the minimal distance between nodes in a partition, introduced in
Section~\ref{sect:ch2-aoi}, as well as the requirement to partition the nodes in a network discussed
in Section~\ref{sect:choose-fake}, give rise to two optimization problems.  We refer to these problems as
the partition problem and the partitioning problem.

For this section alone, for simplicity, we will consider the ``minimal pairwise distance'' definition of $\sigma$
presented in Section~\ref{sect:ch2-aoi}.
That is, in both problems, you are given a set of distinct points $S$ in $\mathbb{R}^2$, with \mbox{$|S| \ge 2$}.
We define the ``goodness'' of any subset \mbox{$T \subseteq S$} as
\begin{equation*}
\proc{G}(T) = \min\left\{\proc{D}(x, y)~|~a, b \in T, a \neq b\right\}\;,
\end{equation*}
where \proc{D} is Euclidean distance in $\mathbb{R}^2$.
In other words, \mbox{$\proc{G}(T)$} is the minimal Euclidean distance between any two points in $T$.
We also extend the notion of goodness to multiple subsets as
\begin{equation*}
\proc{G}(T_1, \ldots, T_n) = \min\left\{\proc{G}(T_1), \ldots, \proc{G}(T_n)\right\}\;.
\end{equation*}
Or, in other words, the goodness of a set of subsets is only as good as the goodness of the worst subset.

In the partition optimization problem, you are given a set of points $S$ in $\mathbb{R}^2$
and an integer \mbox{$2 \le k \le |S|$}.  The goal is to find the subset $T$ of $S$ with
\mbox{$|T| = k$} that maximizes \mbox{$\proc{G}(T)$}.  This problem is equivalent to maximizing the
area of interest (AOI) privacy in a single set of queries to $k$ nodes (under the
``minimal pairwise distance'' definition of $\sigma$).

In the partitioning optimization problem, you are given a set of distinct points $S$ in $\mathbb{R}^2$
and two integers $k_\mathit{min}$ and $k_\mathit{max}$ such that \mbox{$2 \le k_\mathit{min} \le k_\mathit{max} \le |S|$}.  The goal is to
partition $S$ into subsets \mbox{$T_1, \ldots, T_m$} with \mbox{$k_\mathit{min} \le |T_i| \le k_\mathit{max}$}
for all $i$, such that \mbox{$\proc{G}(T_1, \ldots, T_m)$} is maximized.
Here, partitioning into subsets has the usual meaning:
\mbox{$T_i \cap T_j = \emptyset$} for all \mbox{$i \neq j$}, and \mbox{$\cup_i T_i = S$}.
This problem is equivalent to maximizing the AOI-privacy over all the partitions in the partitioning
of a network (again, under the ``minimal pairwise distance'' definition of $\sigma$).

We begin with decision variants of these two problems on a graph (as opposed to in $\mathbb{R}^2$):
\begin{itemize}
\item {\sc Partition-Graph}: You are given a graph \mbox{$G = \langle V, E\rangle$},
\mbox{$k \in \mathbb{Z} \ge 2$}, and \mbox{$d \in \mathbb{Z} \ge 1$}.
Is there a subset \mbox{$S \subseteq V$} with \mbox{$|S|=k$}
such that the shortest path between any two elements in $S$ has length at least $d$?
\item {\sc Partitioning-Graph}: You are given a graph \mbox{$G = \langle V, E\rangle$},
\mbox{$k_\mathit{min} \in \mathbb{Z}$} and \mbox{$k_\mathit{max} \in \mathbb{Z}$}
with \mbox{$1 \le k_\mathit{min} \le k_\mathit{max} \le |V|$},
and \mbox{$d \in \mathbb{Z} \ge 1$}.
Can the vertices of $G$ be partitioned into disjoint sets \mbox{$T_1, \ldots, T_m$} such that
\mbox{$k_\mathit{min} \le |T_i| \le k_\mathit{max}$} and the shortest path between any two vertices in
$T_i$ has length at least $d$, for all $i$?
\end{itemize}

In Appendix \ref{hardness} we prove that both problems are NP-complete. Furthermore, the partition optimization problem is known to be NP-hard on $\mathbb{R}^2$ as well as on graphs~\cite{wang:geometric}. We additionally conjecture the following: 
\begin{conjecture}
The partitioning optimization problem is NP-hard on $\mathbb{R}^2$.
\end{conjecture}
While we do not have a proof that this problem is NP-hard on $\mathbb{R}^2$,
we also have no reason to doubt as much.  The related problem of partitioning points
on $\mathbb{Z}^2$ into a given number of partitions (with no size restrictions)
such that the points within
any given partition are a minimum distance apart is believed to be NP-hard~\cite{johnson:complete},
and another related problem of painting colours onto aircraft to maximize the distance between similar
colours is also believed to be hard~\cite{roley:usaf}.  A formal proof that the partitioning
problem is NP-hard on $\mathbb{R}^2$ remains an open problem.

The partition optimization problem can be solved in \mbox{$\mathcal{O}(n \log n)$} time,
where \mbox{$|\mathcal{S}| = n$},
if we restrict
the problem to the case of \mbox{$k = 2$}~\cite{jaromczyk:communication}.
The solution is based on performing a Graham scan~\cite{graham:convex} to find the convex hull
of the points $\mathcal{S}$ and determining the two furthest points from each other on the
hull~\cite{toussaint:caliphers}.

The partition optimization problem, for any \emph{fixed} value of $k$, can be solved in polynomial time.
That is because the number
of candidate solutions is \mbox{$\mathcal{O}(n^k)$},
where \mbox{$|\mathcal{S}| = n$},
which is polynomial for a fixed $k$.

\subsection{Greedy Partition Approximation}\label{sect:greedyApprox}

We present here a \mbox{$\Theta(1)$} greedy approximation algorithm to the partition problem, based on
private communication with Har-Peled and Sidiropoulos~\cite{communication:har-peled}:
\begin{enumerate}
\item Pick $x_1$ to be an arbitrary node in the network; then,
\item Inductively, given \mbox{$X_{i} = \{x_1, \ldots, x_{i}\}$}, \mbox{$i < k$}, pick
node \mbox{$x_{i+1} \notin X_{i}$} such that
\mbox{$\min\left(\proc{D}(x_1, x_{i+1}), \ldots, \proc{D}(x_i, x_{i+1})\right)$} is maximized.
\end{enumerate}
The algorithm is a simple loop that, in each iteration, considers \mbox{$\Theta(n)$} nodes to
be the next possible node added to the partition~---~doing so would require only a single distance
computation between the node under consideration and the most recent node added to the partition,
possibly causing the minimal distance to any node in the partition to be updated for the node under
consideration.  The loop iterates \mbox{$k-1$} times, so clearly this algorithm can be implemented
in \mbox{$\mathcal{O}(nk)$} time with \mbox{$\mathcal{O}(n)$} memory.

This approximation algorithm is similar to one proposed by Ravi~et~al.~\cite{ravi:heuristic}
for the maximal minimum facility dispersion problem.  A version for facility dispersion
on continuous graphs was proposed by Tamir~\cite{tamir:obnoxious}.
For the convenience of the reader, we have included an original proof
that the algorithm in indeed a \mbox{$\Theta(1)$} approximation~---~in
fact, a 2-approximation~---~for the partition problem.

\begin{lemma}\label{lemma:greedy-partition}
Let \mbox{$X = \{x_1, \ldots, x_k\}$} be the partition returned by the greedy algorithm, with
minimal pairwise distance $d^\prime$.
Let \mbox{$Y = \{y_1, \ldots, y_k\}$} be an optimal solution to the partition problem.
Let \mbox{$1 < s \le k$} be the smallest integer such that there exists \mbox{$t < s$} where
\mbox{$\proc{D}(x_s, x_t) = d^\prime$}.
Then, for all \mbox{$y \in Y$}, there exists an
\mbox{$x_y \in X_{s-1} = \{x_1, \ldots, x_{s-1}\}$} such that \mbox{$\proc{D}(y, x_y) \le d^\prime$}.
\begin{proof}
Let $y$ be an arbitrary node of the optimal solution $Y$ with a distance
greater than $d^\prime$ from every node in $X_{s-1}$.  However, the node chosen by
the greedy algorithm as $x_s$ has the greatest minimal distance to any node in
$X_{s-1}$, namely $d^\prime$, which is a contradiction.
\end{proof}
\end{lemma}

\begin{theorem}\label{theorem:2approx}
The greedy algorithm is a 2-approximation for the partition problem.
\end{theorem}

The greedy algorithm and its proof in Theorem~\ref{theorem:2approx} generalize to choosing
points to maximize the minimal pairwise distance over any continuous convex space in $\mathbb{R}^2$.

\subsection{Divide-and-Conquer Partitioning}\label{sect:divideAndConquer}

In this section, we present a divide-and-conquer algorithm for effectively solving the partitioning problem.
The idea underlying the algorithm is first to compute the number of nodes that will be assigned to each
partition within the partitioning.  The algorithm then finds a partition of size $k$, where $k$ is the largest
number of nodes that have to be assigned to any partition (up to a maximum of \mbox{$k = \mu$}).  This
partition could be found as an optimal partition via brute-force, or via a Monte Carlo algorithm (in our implementation,
we use a Monte Carlo algorithm that generates $\pi$ random partitions and chooses the best one).
The nodes in this partition serve as \emph{seeds} that form the basis of subnets of the network.
Each node in the network is assigned to the subnet of its closest seed, and the number of nodes to be
assigned to each partition are divided evenly among the subnets. The algorithm then continues recursively on each
subnet, until the maximal number of nodes to be assigned to each partition in a given subnet is only one~---~at
this point, each node in the subnet is assigned at random to one of the partitions requiring a node from that subnet.
The algorithm is formalized in Algorithm~\ref{alg:divide-and-conquer}, with the entire process repeating $\alpha$
times and the best partitioning out of all of those being returned on line~\ref{algline:returnbest}.  The integer division
that begins the $\Call{Divide-And-Conquer}$ function is just to divide the number of nodes that will be in each partition
as equally as possible.

\begin{algorithm}
\begin{algorithmic}[1]
 \Function{Divide-And-Conquer}{$\mathcal{N}, k_\mathit{min}, k_\mathit{max}, \alpha, \pi, \mu$}
    \State \Preconditions $1 \le k_\mathit{min} \in \mathbb{Z} \le k_\mathit{max} \in \mathbb{Z} \le |\mathcal{N}|, \alpha \in \mathbb{Z} \ge 1, \pi \in \mathbb{Z} \ge 1, \mu \in \mathbb{Z} \ge 2$
    \State $\mathit{numPartitions} \gets |\mathcal{N}| / k_\mathit{min}$ \Comment Integer division
    \If {$\mathit{numPartitions} \cdot k_\mathit{max} < |\mathcal{N}|$}
        \State \Error Impossible to generate partitions
    \EndIf
    \State $\mathit{minSize} \gets |\mathcal{N}| / \mathit{numParititons}$ \Comment Integer division
    \State $\mathit{numLarger} \gets |\mathcal{N}| - \left(\mathit{minSize} \cdot \mathit{numPartition}\right)$
    \State $\mathit{numInPart}[i] \gets \mathit{minSize} + 1~~~\forall~0 \le i < \mathit{numLarger}$ \Comment Arrays indexed from 0
    \State $\mathit{numInPart}[i] \gets \mathit{minSize}~~~\forall~numLarger \le i < \mathit{numPartitions}$
    \For {$i \gets 0 \To \alpha-1$}
        \State Assign $N$ to no partition~$\forall~N \in \mathcal{N}$
        \State \Call{Divide-And-Conquer-X}{$\mathcal{N}, \mathit{numInPart}, \mathit{numPartitions}, \pi, \mu$} 
        \State $p[i] \gets$ the current partitioning
    \EndFor
    \State \Return the best partitioning out of $p[0], \ldots, p[\alpha-1]$\label{algline:returnbest}
\EndFunction
\Statex
\Procedure{Divide-And-Conquer-X}{$\mathcal{N}, \mathit{numInPart}, \mathit{numPartitions}, \pi, \mu$}
    \State $\mathit{numSubnets} \gets \min\left(\mu, \max\left(\mathit{numInPart}[1], \ldots, \mathit{numInPart}[\mathit{numPartitions}]\right)\right)$
    \If {$\mathit{numSubnets} = 1$}
        \ForAll {partitions $P$ with $\mathit{numInPart}[P] = 1$}
            \State Assign a random unassigned node in $\mathcal{N}$ to partition $P$
        \EndFor
        \State \Return
    \EndIf
    \State $\mathit{subnet}[i] \gets \{\}~~\forall~0 \le i < \mathit{numSubnets}$
    \State $\mathit{capacity}[i] \gets 0~~\forall~0 \le i < \mathit{numSubnets}$
    \State $\mathit{nDist}[i][j] \gets 0~~\forall~0 \le i < \mathit{numSubnets}, 0 \le j < \mathit{numPartitions}$
    \State $\mathit{seeds} \gets \Call{Monte-Carlo-Partition}{\mathcal{N}, \mathit{numSubnets}, \pi}$
    \ForAll {$N \in \mathcal{N}$}
        \If {$\exists~i~\text{such that}~\mathit{seeds}[i] = N$}
            \State $\mathit{index} \gets i$
        \Else
            \State $d \gets$ the minimal distance from $N$ to any element in $\mathit{seeds}$
            \State $\mathit{indices} \gets \left\{i~|~\proc{D}\left(N, \mathit{seeds}[i]\right) = d\right\}$
            \State $\mathit{index} \gets$ a uniform random element of $\mathit{indices}$
        \EndIf
        \State Add $N$ to $\mathit{subnet}[\mathit{index}]$
        \State $\mathit{capacity}[\mathit{index}] \gets \mathit{capacity}[\mathit{index}] + 1$
    \EndFor
    \State $\mathit{onSubnet} \gets 0$
    \For {$\mathit{onPartition} \gets 0 \To \mathit{numPartitions}-1$}
        \For {$i \gets 1 \To \mathit{numInPart}[\mathit{onPartition}]$}
            \While {$\mathit{capacity}[\mathit{onSubnet}] = 0$}
                \State $\mathit{onSubnet} \gets \left(\mathit{onSubnet} + 1 \right) \mod \mathit{numSubnets}$
            \EndWhile
            \State $\mathit{nDist}[\mathit{onSubnet}][\mathit{onPartition}] \gets \mathit{nDist}[\mathit{onSubnet}][\mathit{onPartition}] + 1$
            \State $\mathit{capacity}[\mathit{onSubnet}] \gets \mathit{capacity}[\mathit{onSubnet}] - 1$
            \State $\mathit{onSubnet} \gets \left(\mathit{onSubnet} + 1 \right) \mod \mathit{numSubnets}$
        \EndFor
    \EndFor
    \For {$i \gets 0 \To \mathit{numSubnets} - 1$}
        \State \Call{Divide-And-Conquer-X}{$\mathit{subnet}[i], \mathit{nDist}[i], \mathit{numPartitions}, \pi, \mu$}
    \EndFor
\EndProcedure
\Statex
\Function{Monte-Carlo-Partition}{$\mathcal{N}, k, \pi$}
    \State \Return the best partition of size $k$ on $\mathcal{N}$, from $\pi$ randomly generated partitions
\EndFunction
\end{algorithmic}
\caption{The divide-and-conquer algorithm for solving the partitioning problem}\label{alg:divide-and-conquer}
\end{algorithm}

First, we compared the divide-and-conquer algorithm to a randomized partitioning algorithm and to a brute-force
search that finds an optimal partitioning.  The randomized algorithm is formalized as
Algorithm~\ref{alg:random-partitioning}.  It uses the same integer division technique to equally divide the partition
sizes as the divide-and-conquer algorithm, starting on line~\ref{algline:intdiv}; however, once the partition sizes are
equally distributed, it assigns random unassigned nodes to each partition to fill that partition, on
line~\ref{algline:random}.  The entire process of creating a partitioning randomly is repeated $\alpha$ times,
as with the divide-and-conquer algorithm, and the best partitioning is returned.
For the randomized algorithm, we ran 1000 trials using \mbox{$\alpha = 1000$}.  The same number
of trials were used for the divide-and-conquer algorithm, with \mbox{$\alpha = 1000$}, \mbox{$\pi = 10$}, and
\mbox{$\mu = 3$}.  Fewer trials were needed for the brute-force approach; we ran 25 trials just to get an accurate
timing.  The experiment was performed on a 9x5 grid of nodes, and the results are summarized in
Figure~\ref{fig:partition-three-algs}.
The running times, measured on a 2.93~GHz Core~i7 with 8~GB of RAM running OS~X~10.8.4, illustrate that
the divide-and-conquer partitioning algorithm is orders of magnitude faster
than a brute-force search.  While it is slower than a randomized approach, the divide-and-conquer algorithm
also returned significantly better results than the random algorithm.

\begin{algorithm}
\begin{algorithmic}[1]
\Function{Random-Partitioning}{$\mathcal{N}, k_\mathit{min}, k_\mathit{max}, \alpha$}
    \State \Preconditions $1 \le k_\mathit{min} \in \mathbb{Z} \le k_\mathit{max} \in \mathbb{Z} \le |\mathcal{N}|, \alpha \in \mathbb{Z} \ge 1$
    \State $\mathit{numPartitions} \gets |\mathcal{N}| / k_\mathit{min}$ \Comment Integer division\label{algline:intdiv}
    \If {$\mathit{numPartitions} \cdot k_\mathit{max} < |\mathcal{N}|$}
        \State \Error Impossible to generate partitions
    \EndIf
    \State $\mathit{minSize} \gets |\mathcal{N}| / \mathit{numParititons}$ \Comment Integer division
    \State $\mathit{numLarger} \gets |\mathcal{N}| - \left(\mathit{minSize} \cdot \mathit{numPartition}\right)$
    \State $\mathit{numInPart}[i] \gets \mathit{minSize} + 1~~~\forall~0 \le i < \mathit{numLarger}$ \Comment Arrays indexed from 0
    \State $\mathit{numInPart}[i] \gets \mathit{minSize}~~~\forall~numLarger \le i < \mathit{numPartitions}$
    \For {$i \gets 1 \To \alpha$}
        \State Assign $N$ to no partition~$\forall~N \in \mathcal{N}$
        \State \Call{Random-Partitioning-X}{$\mathcal{N}, \mathit{numInPart}, \mathit{numPartitions}$} 
        \State $p[i] \gets$ the current partitioning
    \EndFor
    \State \Return the best partitioning out of $p[0], \ldots, p[\alpha-1]$
\EndFunction
\Statex
\Procedure{Random-Partitioning-X}{$\mathcal{N}, \mathit{numInPart}, \mathit{numPartitions}$}
    \For{$\mathit{onPartition} \gets 0 \To \mathit{numPartitions}-1$}
        \For {$i \gets 1 \To \mathit{numInPart}[\mathit{onPartition}]$}
            \State Assign a random unassigned node in $\mathcal{N}$ to partition $\mathit{onPartition}$\label{algline:random}
        \EndFor
    \EndFor
\EndProcedure
\end{algorithmic}
\caption{The random algorithm for solving the partitioning problem}\label{alg:random-partitioning}
\end{algorithm}

\begin{figure*}[h] 
\centering 
\begin{subfigure}[b]{0.48\textwidth}
\centering
\includegraphics[width=\textwidth]{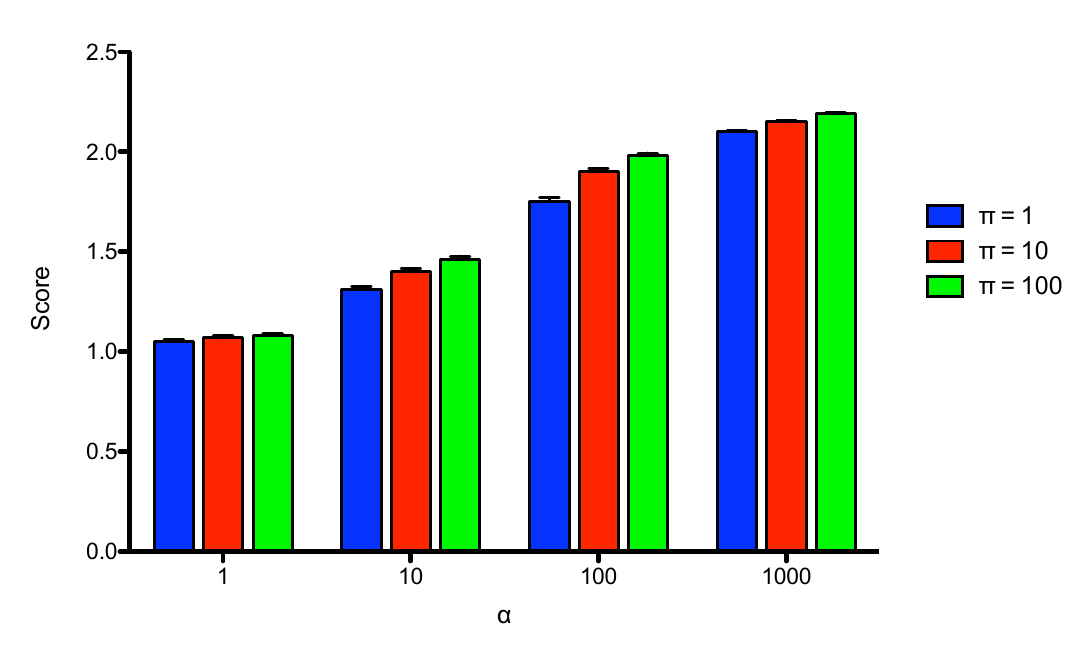}
\end{subfigure}
\hfill
\begin{subfigure}[b]{0.48\textwidth}
\centering
\includegraphics[width=\textwidth]{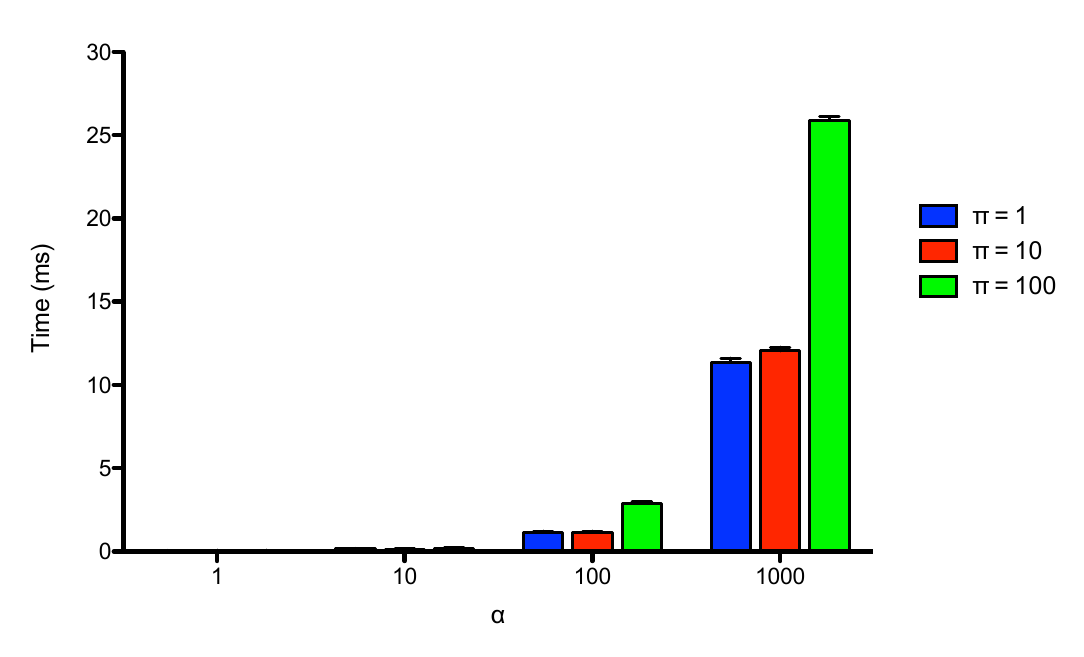}
\end{subfigure}
\caption{The effect of the $\alpha$ and $\pi$ parameters on the score and running time of the divide-and-conquer partitioning algorithm on a grid topology.} 
\label{fig:partitioning-attempts} 
\end{figure*}

\begin{figure*} 
\centering 
\begin{subfigure}[b]{0.48\textwidth}
\centering
\includegraphics[width=\textwidth]{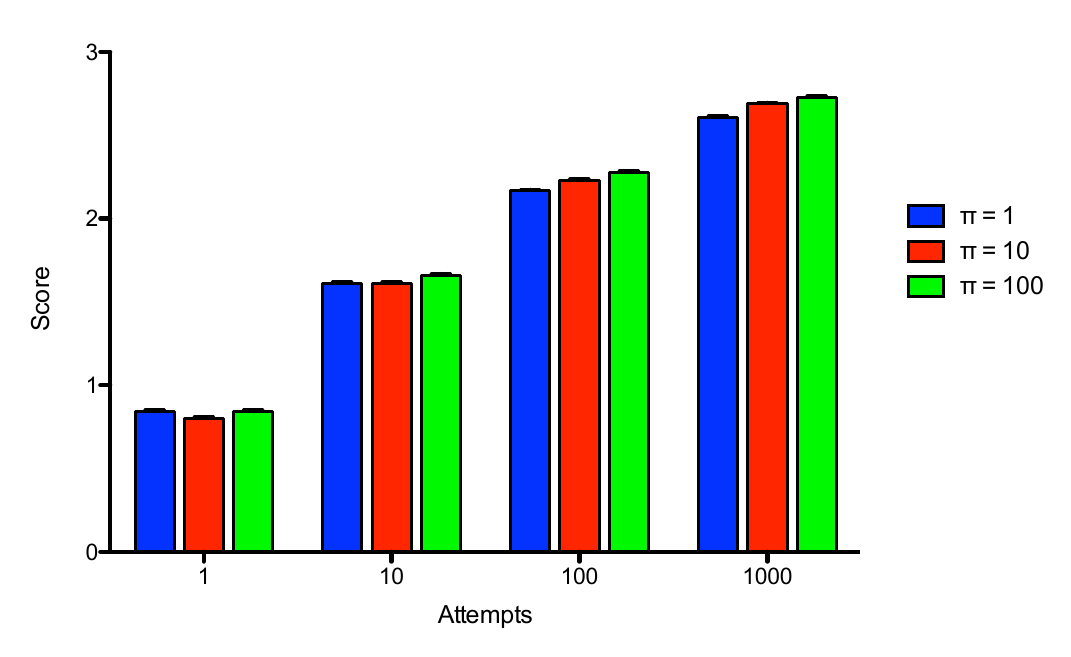}
\end{subfigure}
\hfill
\begin{subfigure}[b]{0.48\textwidth}
\centering
\includegraphics[width=\textwidth]{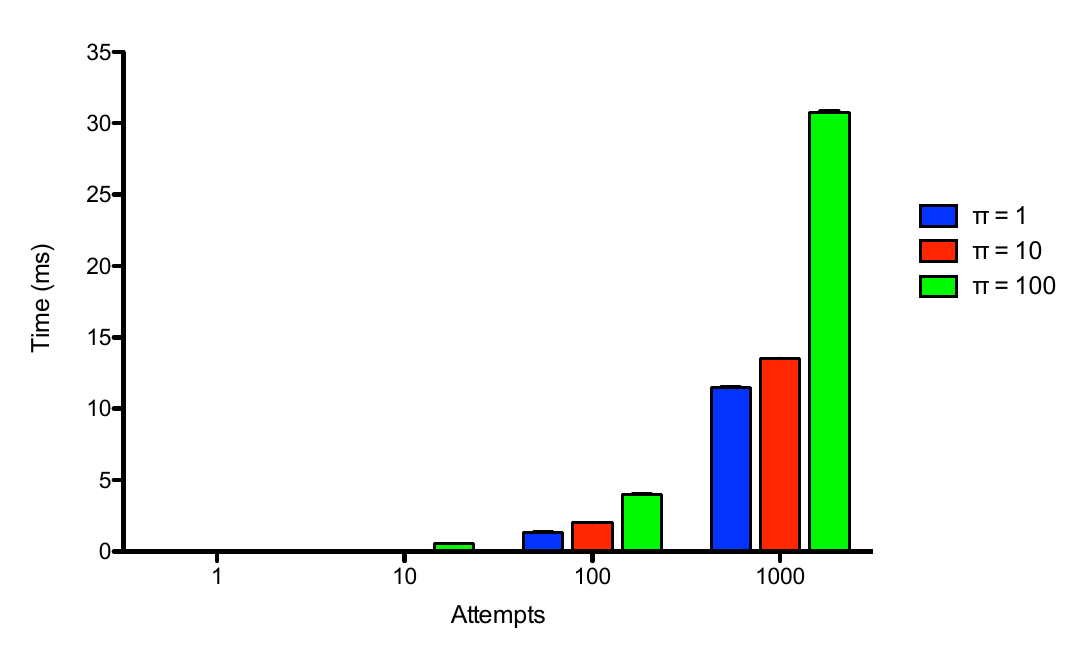}
\end{subfigure}
\caption{The effect of the $\alpha$ and $\pi$ parameters on the score and running time of the divide-and-conquer partitioning algorithm on a randomized topology.} 
\label{fig:partitioning-attempts-random} 
\end{figure*}

The effect of the variables $\alpha$ (the number of Monte Carlo attempts run for the entire partitioning algorithm)
and $\pi$ (the number of partitions generated each time seeds are needed) is highlighted in
Figure~\ref{fig:partitioning-attempts}.  We modified the values of $\alpha$ and $\pi$, while running the
divide-and-conquer partitioning algorithm on the same 9x5 grid.  For each \mbox{$\langle\alpha,\pi\rangle$}
pair, we ran 1000 trials.  Unsurprisingly, increasing the values of $\alpha$ and $\pi$ increases the
average quality of the partitioning that is chosen in the end.  However, as evidenced by the running time graph,
increasing $\pi$ to large values can increase the running time significantly, for comparably little payoff in the average
quality of the partitioning found (relative to increasing $\alpha$).

We confirmed these results by re-running the experiment, but generating a new random topology for
each trial.  Each random topology was generated by placing nodes uniformly at random in a 10x10 area,
measured in an arbitrary unit of length.  To account for the variance in topology, we increased the number
of trials for each \mbox{$\langle\alpha,\pi\rangle$} pair to 5000.  The results on the random topology,
illustrated in Figure~\ref{fig:partitioning-attempts-random}, are comparable
to those seen on the grid topology.  Based on these results, we recommend using a value of \mbox{$\pi = 10$}, with
the $\alpha$ value based on both the size of the network and the necessity of a higher-scoring result (we chose
to use \mbox{$\alpha = 1000$} for the remainder of this section).

\begin{figure}[htbp]
  \centering
  \begin{minipage}{0.48\textwidth}
    \centering
    \includegraphics[width=\linewidth]{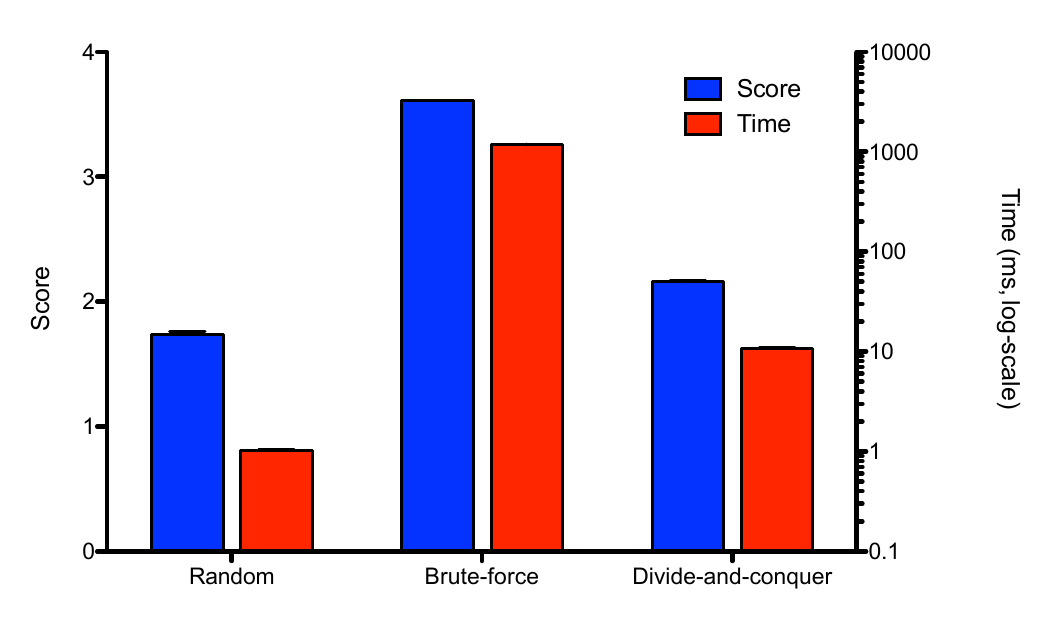}
    \caption{A comparison of the scores of the partitioning computed by three different partitioning algorithms, and the running time of each.}\label{fig:partition-three-algs}
  \end{minipage}
  \hfill
  \begin{minipage}{0.48\textwidth}
    \centering
    \includegraphics[width=\linewidth]{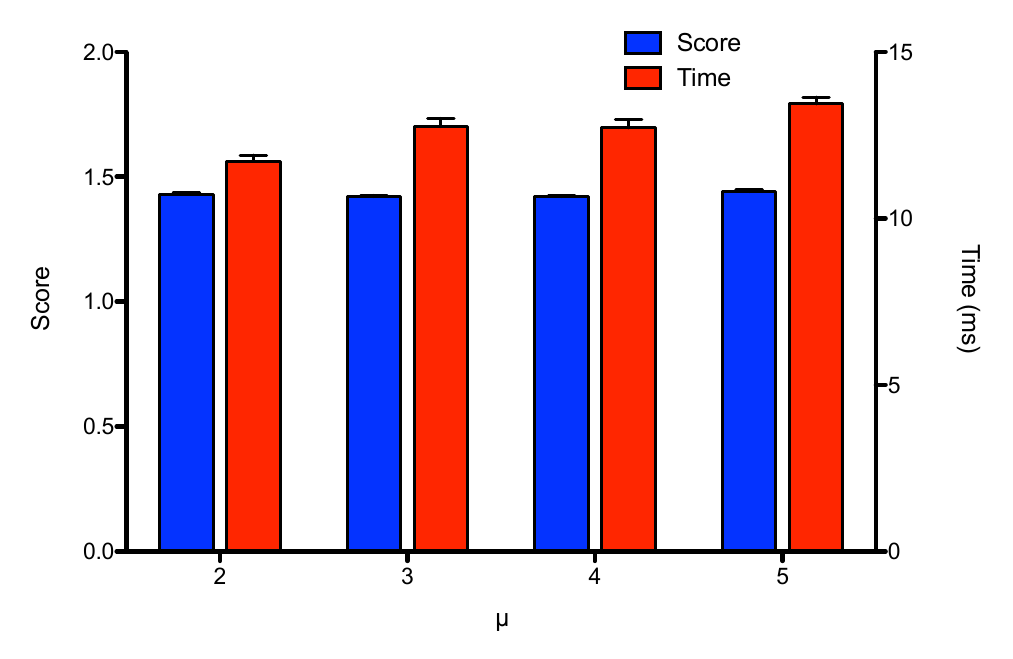}
    \caption{The effect of $\mu$, the maximal partition seed size, on the final score and running time of the divide-and-conquer partitioning algorithm.}\label{fig:partitioning-p-size}
  \end{minipage}
\end{figure}

The $\mu$ parameter is simply a safeguard for when the individual partitions in a partitioning contain a large
number of nodes.  If this were the case, the size of the partition used as a seed would be very large.
If the seed partition is sufficiently large, relative to the size of the network
(i.e., if there are a small number of large partitions),
the divide-and-conquer algorithm would nearly degenerate into a randomized partitioning.
In most cases, the value of $\mu$ can be safely ignored (i.e., set \mbox{$\mu = \infty$}).
As illustrated in Figure~\ref{fig:partitioning-p-size}, the $\mu$ parameter had no significant effect
on the final score, and the individual trial run times are so variable that $\mu$ had no overall significant effect
on the running time of the algorithm.  We recommend using \mbox{$\mu = 3$} for safety's sake; however,
if the usage case precludes large individual partitions, the $\mu$ parameter can be entirely ignored.

\section{Experimental Results in a Uniform Environment}\label{sect:ch2-results}

We begin our experiments on the $k$-anonymity scheme in a featureless environment, devoid of any
statistics learned from observing human mobility.  The network is a 10x10 grid of nodes placed unit-distance
apart, and for all nodes $N$ and $M$ in the network \mbox{$\probC{N} = \probL{M}{N} = \frac{1}{100}$} .
The user's current location is chosen proportional to $\probC{N}$ (i.e., uniformly at random, for now),
the user's real location of interest to be queried is chosen proportional to
$\probL{M}{N}$ (i.e., also uniformly at random for the moment), and the \mbox{$k-1$} dummy queried locations
are always chosen uniformly at random.

Since each trial of an experiment will observe the effect on the user's CL-privacy and LOI-privacy from a
single set of $k$ queries being issued, we do not need to generate a partitioning of the network in the
implementation of the experiment.  However, it should be noted that this behaviour is equivalent to
generating a uniform random partitioning, then querying a partition of size $k$.  Because we wish to observe
correlations between, e.g., AOI-privacy and CL-privacy, using a random partitioning scheme (as opposed
to the partitioning scheme discussed in Section~\ref{sect:divideAndConquer}) is preferred because
it gives us the widest range of results.

In our first experiment, message headers between nodes are encrypted using a global key that has been
compromised by the adversary.  The adversary has compromised 10 of the nodes in the network, and is using
them to eavesdrop on communications.  The user issues \mbox{$k = 3$} queries into the network from
their current location, and the adversary attempts to determine the user's current location and real location of
interest.

Figure~\ref{fig:scatter-loi} plots the user's resulting CL-privacy and LOI-privacy from issuing a single set of $k$
queries~---~each point on the scatterplot represents one of the 1000 trials performed.
Unsurprisingly, there is a significant correlation (\mbox{$P < 0.0001$})  between the user's resulting
LOI-privacy and CL-privacy, in that high CL-privacy is essential to achieving high LOI-privacy.
After all, the more information the adversary gains from overhearing some or all of a set of $k$ queries,
the better able the adversary is to determine both the user's current location and location of interest.
Interestingly, LOI-privacy in the featureless environment is essentially an ``all or nothing'' measure,
in that the user achieves near-perfect LOI-privacy, i.e., approximately
\begin{equation*}
\frac{|\mathcal{N}_\mathit{honest}| - 1}{|\mathcal{N}_\mathit{honest}|} = \frac{89}{90}\;,
\end{equation*}
or near-minimal LOI-privacy, i.e., approximately
\begin{equation*}
\frac{k-1}{k} = \frac{2}{3}\;.
\end{equation*}
However, the scatterplot is somewhat misleading: for example, the cluster
of points in the upper-right of the graph do not all have LOI-privacy of $\frac{89}{90}$;
many have LOI-privacy slightly below that perfect measure.  That is, the adversary was always able to
deduce, even in tiny measure, locations that could not have been queried.  However, in the featureless
environment, there is no information for the adversary to work with regarding location popularities,
significantly handicapping the adversary who is unable to overhear or deduce all of the queried locations.

\begin{wrapfigure}{L}{0.48\textwidth}
    \centering
    \includegraphics[width=\linewidth]{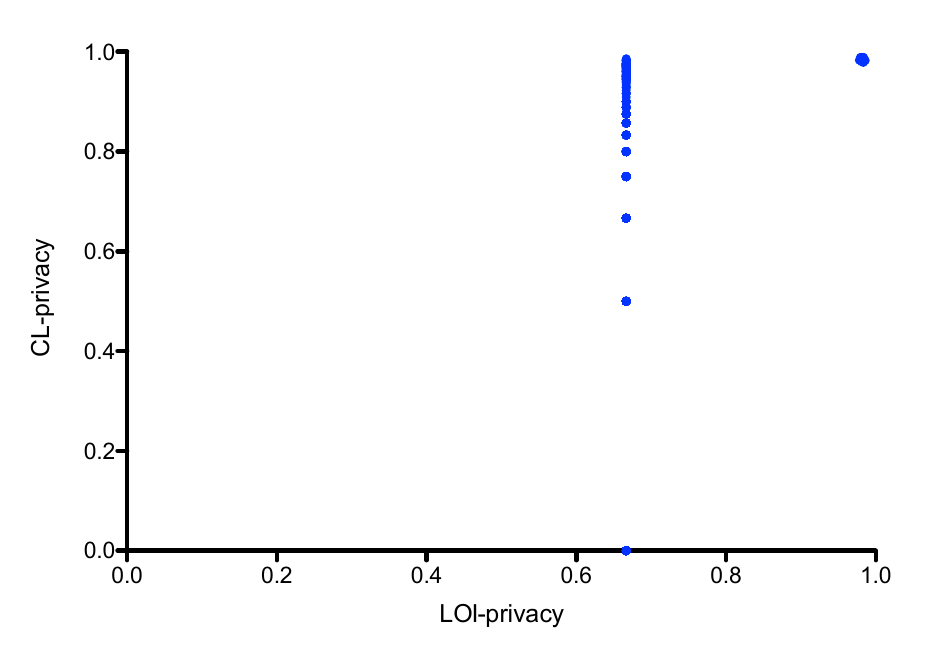}
    \caption{The user's CL-privacy and LOI-privacy over 1000 trials with $k=3$.}\label{fig:scatter-loi}
\end{wrapfigure}

Of more interest to a user, however, might be what changes they can make in order to increase the privacy
of their current location and location of interest.  The most obvious change that the user can make is to adjust
$k$, the number of queries that they send.  We repeated the previous experiment, this time adjusting
\mbox{$k \in [2, 5]$}.
As illustrated in Figure~\ref{fig:static-k-privacy},
lower values of $k$ resulted in significantly higher CL-privacy, as the adversary has fewer
query message with which to ``triangulate'' the user.  On the other hand, higher values of $k$ resulted in significantly
higher LOI-privacy, as it becomes less likely that the adversary will guess the correct queried location as the user's
real location of interest.  Each increment of $k$ resulted in a significant change to both privacy metrics
(\mbox{$P < 0.05$}), as determined
by a one-way ANOVA test with Bonferroni's multiple comparison test~\cite{graphpad:bonferroni}
on each pair of incremental $k$ values (i.e.,
from \mbox{$k=2$} to \mbox{$k=3$}, from \mbox{$k=3$} to \mbox{$k=4$}, and from \mbox{$k=4$} to \mbox{$k=5$}),
thus indicating a clear trade-off between the user's ability to protect the privacy of their current location and the
privacy of their location of interest.

\begin{wrapfigure}{R}{0.48\textwidth}
    \centering
    \includegraphics[width=\linewidth]{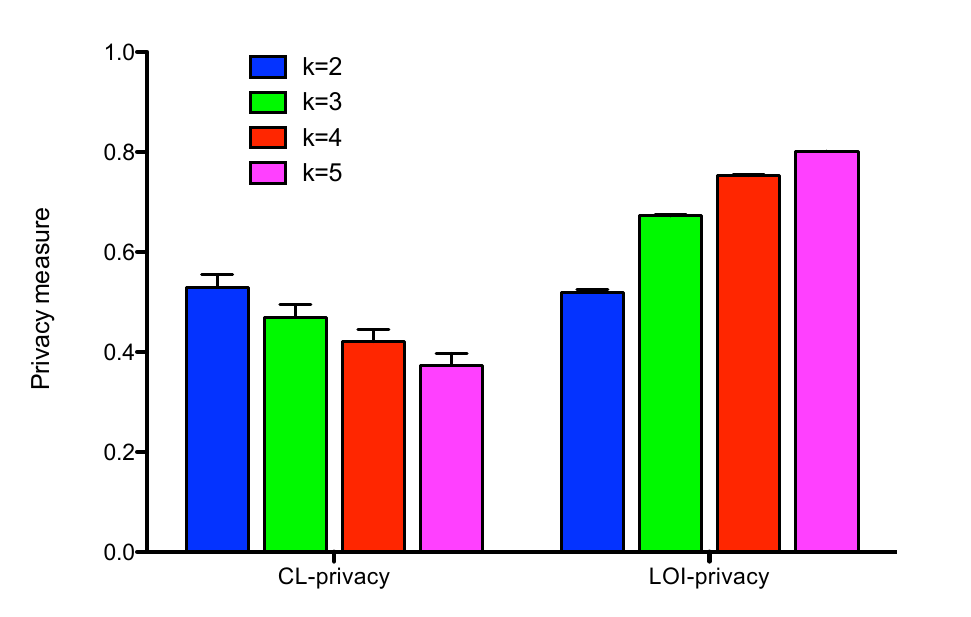}
    \caption{The user's CL-privacy and LOI-privacy, as the number of queries $k$ is varied, over 1000 trials.}\label{fig:static-k-privacy}
\end{wrapfigure}

Another parameter within the user's control is their AOI-privacy, which is determined by
how well-dispersed their $k$ queries are physically around the network.  Figure~\ref{fig:scatter-aoi} shows
the user's CL- and LOI-privacy, relative to their chosen AOI-privacy over 1000 trials.
While there is a significant correlation (\mbox{$P < 0.0001$}) between
between the user's AOI-privacy and their LOI-privacy, it is not
large (\mbox{$R^2 = 0.0158$}).  While that may seem surprising, given that the two are
conceptually similar measures of how well protected the identity of the user's real query destination is, that
lack of large correlation just indicates that how physically spread out the $k$ destinations are is not the
most important factor in helping the adversary narrow down the real location of interest.
There is also no large correlation (\mbox{$P < 0.01$}, \mbox{$R^2 < 0.01$}) between a user's
AOI-privacy and their CL-privacy.  Intuitively, one might think that the more dispersed the $k$ queries are,
the more opportunity there would be for the adversary to overhear queries and triangulate the user's location.
While that is probabilistically true, as indicated by the low $P$ value, it is not a largely meaningful factor to
the adversary, as indicated by the low $R^2$ value.

\begin{figure*} 
\centering 
\begin{subfigure}[b]{0.48\textwidth}
\centering
\includegraphics[width=\textwidth]{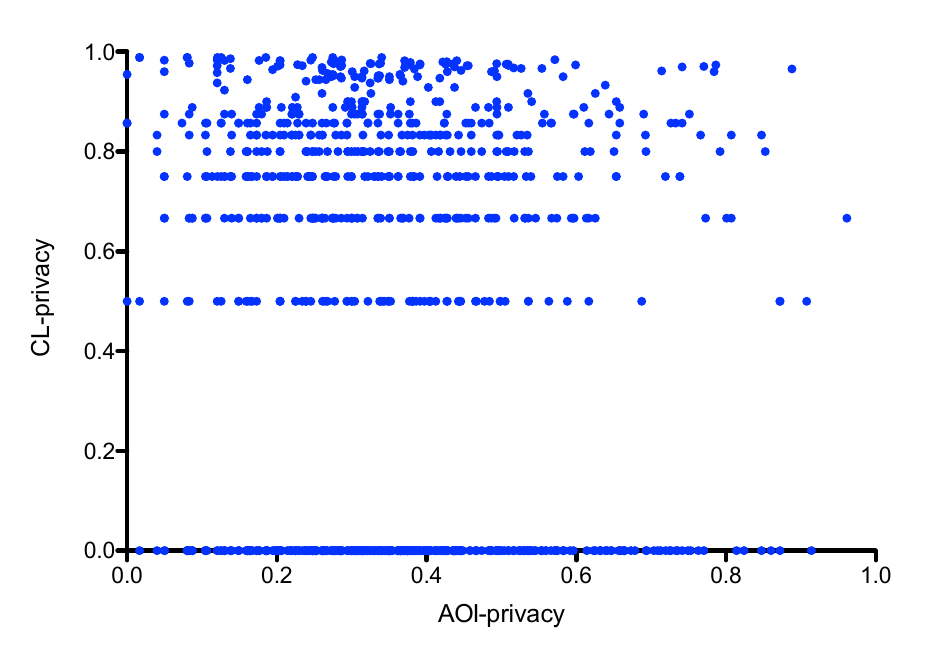}
\end{subfigure}
\hfill
\begin{subfigure}[b]{0.48\textwidth}
\centering
\includegraphics[width=\textwidth]{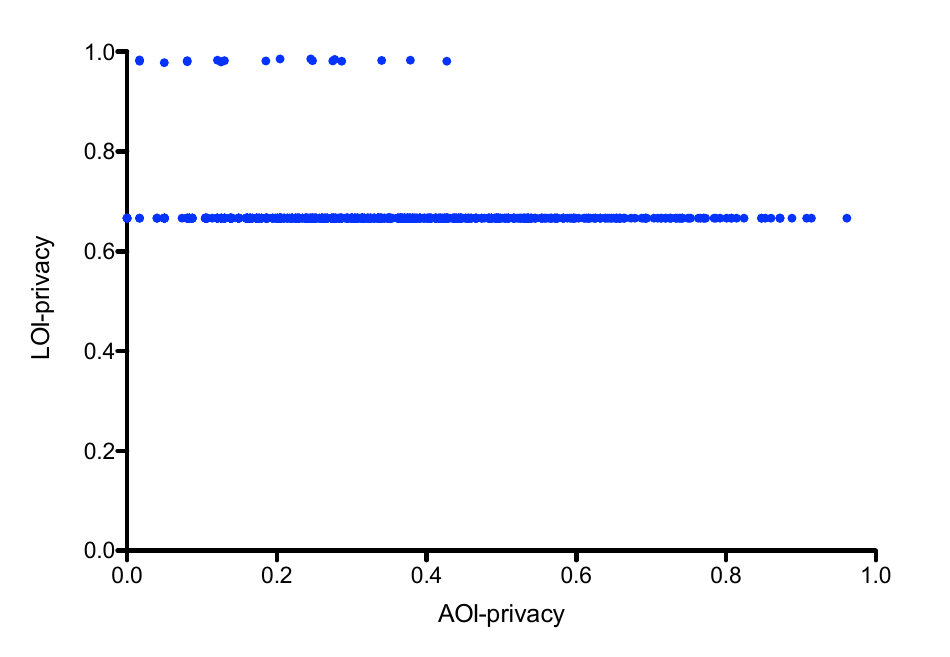}
\end{subfigure}
\caption{The user's CL-privacy and LOI-privacy, as the user's AOI-privacy is varied, over 1000 trials with \mbox{$k=3$}.} \label{fig:scatter-aoi} 
\end{figure*}

To confirm these results, we studied a metric related to AOI-privacy, more closely related to the
number of opportunities the adversary will have to intercept user queries.
That metric is the number of unique hops the user's $k$ queries
take (i.e., the number of query message transmissions between unique pairs of nodes),
which the user is able to control based on the chosen nodes to query.
Figure~\ref{fig:scatter-hops} shows the user's CL- and LOI-privacy, relative to the number of distinct hops taken
by their query messages over 1000 trials.
Again, the number of distinct hops taken by the user's queries is significantly
negatively correlated with both the user's resulting CL-privacy (\mbox{$P = 0.0097$}) and their
resulting LOI-privacy (\mbox{$P < 0.0001$}).  However, once again, the $R^2$ values are low:
\mbox{$R^2 = 0.0067$} and \mbox{$R^2 = 0.0318$} respectively.
In other words, the user's control over how dispersed their queries are will, probabilistically, influence
their CL- and LOI-privacy, but blind luck based on query routing near compromised nodes is a much
more important factor.

\begin{figure*} 
\centering 
\begin{subfigure}[b]{0.48\textwidth}
\centering
\includegraphics[width=\textwidth]{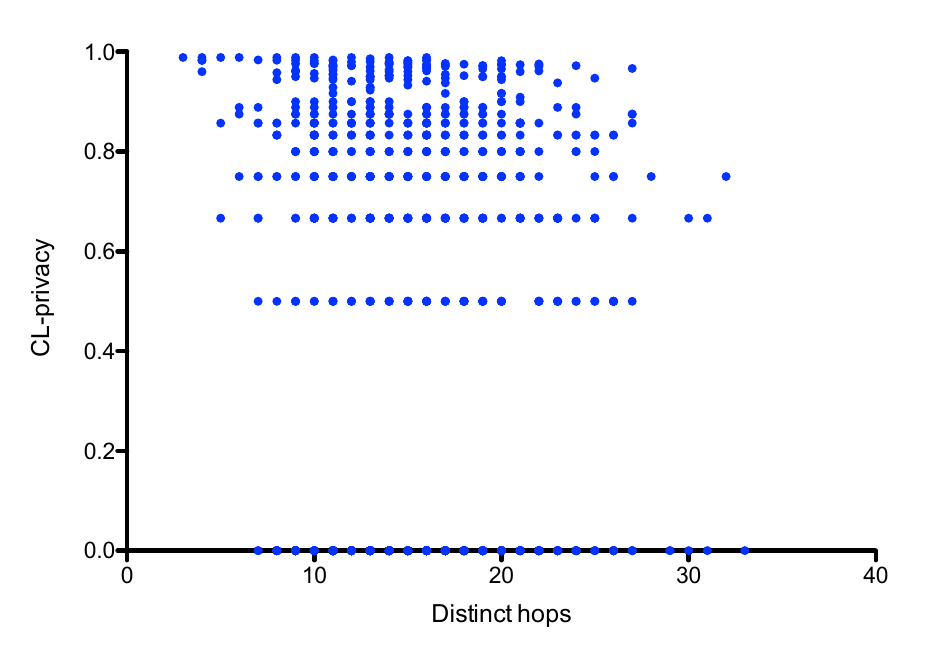}
\end{subfigure}
\hfill
\begin{subfigure}[b]{0.48\textwidth}
\centering
\includegraphics[width=\textwidth]{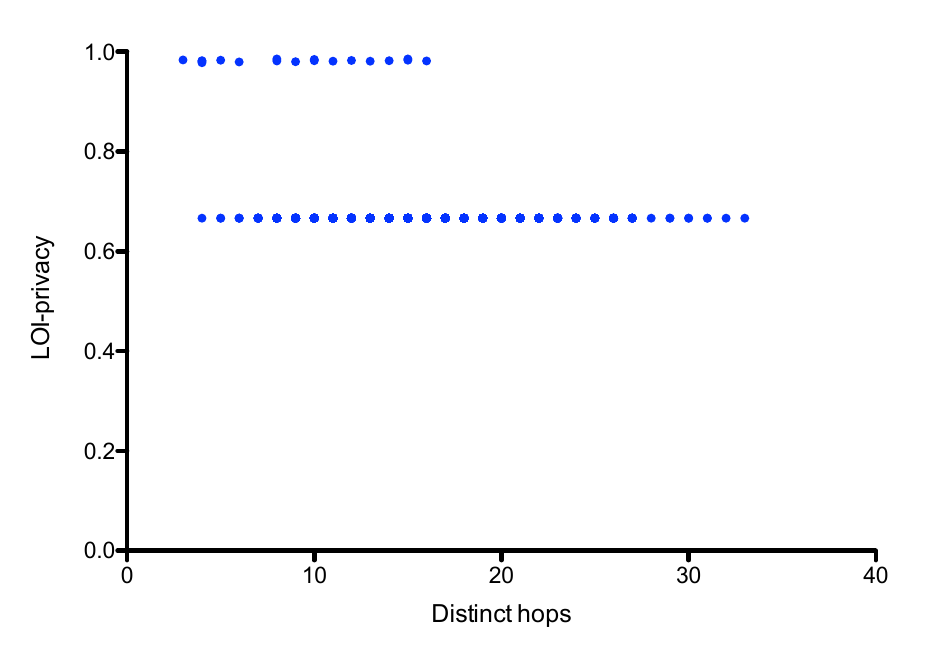}
\end{subfigure}
\caption{The user's CL-privacy and LOI-privacy, as the number of distinct hops in the queries issued by the user is varied, over 1000 trials with \mbox{$k=3$}.} 
\label{fig:scatter-hops} 
\end{figure*}

There are numerous other factors, potentially beyond the control of users of the network, that affect the adversary's
success as well.  One critical factor for the adversary is how many nodes are compromised in the network.  The
more nodes that are compromised, the more likely it is that the adversary will be able to intercept queries, learning
the locations queried and facilitating the computation of the user's current location.  This result is depicted in
Figure~\ref{fig:static-compromised-results}, showing the predictable result that increasing the number of nodes
compromised by the adversary decreases both the user's CL- and LOI-privacy.  Note, however, that
Figure~\ref{fig:static-compromised-results} illustrates the same trend first illustrated in
Figure~\ref{fig:static-k-privacy}: increasing the privacy constant $k$~---~a constant
that is entirely under the control of the user~---~decreases the user's CL-privacy
but increases the user's LOI-privacy, illustrating the trade-off between a user protecting the privacy of their
current location and that of their location of interest.

\begin{figure*} 
\centering 
\begin{subfigure}[b]{0.48\textwidth}
\centering
\includegraphics[width=\textwidth]{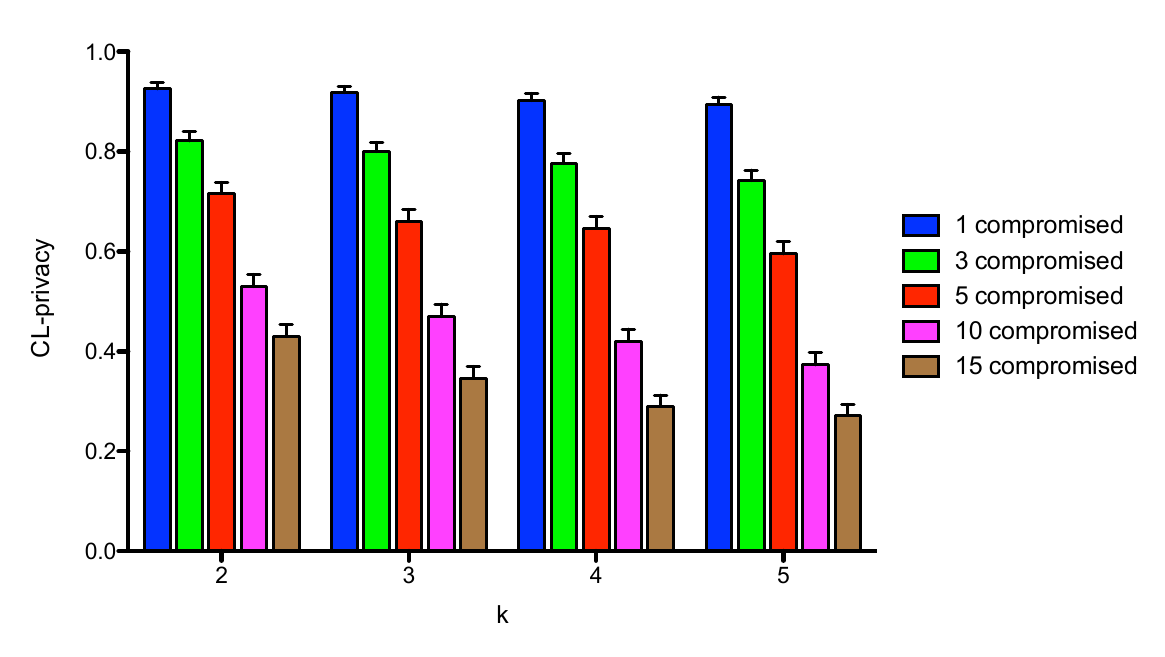}
\end{subfigure}
\hfill
\begin{subfigure}[b]{0.48\textwidth}
\centering
\includegraphics[width=\textwidth]{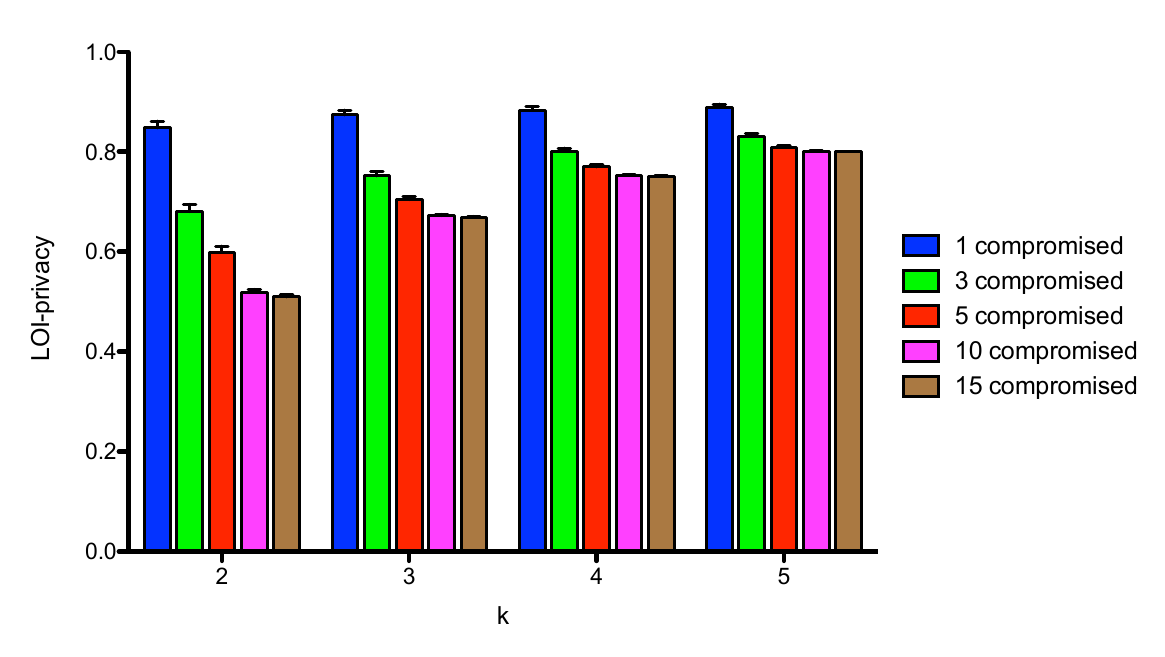}
\end{subfigure}
\caption{The user's CL-privacy and LOI-privacy, as the number of nodes compromised by the adversary varies, over 1000 trials.} 
\label{fig:static-compromised-results} 
\end{figure*}

A property that may or may not be under the control of the user is the type of encryption used in communication
between nodes.  Recall that the query proper can be encrypted end-to-end between the user and the queried
node using any number of available cryptographic techniques.  The query itself is not at issue.
The issue at hand is what information the adversary can glean from decrypting the information contained in the
headers of overheard messages.  So far, we have assumed that transmissions between the nodes in the
network are encrypted using a global key.  Should any node be compromised, the adversary learns the global key;
then, if any query transmission is overheard, the adversary learns the destination of that query, as well as
two nodes along the path to that destination (the transmitter and the intended recipient of the transmission).
However, it is possible that transmissions between nodes are encrypted using link keys~---~that is,
every pair of nodes in transmission range of each other share a unique symmetric key.  In this case,
an overheard query per se reveals to the adversary only a single node~---~the transmitter~---~in
the path from the user's current location to the queried location (we must assume that the adversary can
determine which node transmitted the query based on metrics such as signal strength, timing,
etc.~\cite{bouassida:sybil}).  It
is only if a compromised node is the recipient of the transmission that the adversary learns the ultimate
destination of the query (and, obviously, the recipient of the transmission).  A comparison of the user's CL-
and LOI-privacy based on whether global keys or link keys are used is illustrated
in Figure~\ref{fig:static-adversary-results}; for the purposes of comparison, an environment in which there is no
broadcast transmission (e.g., optical or wired communication~\cite{yick:survey}),
and thus the adversary learns nothing unless a compromised node is the recipient or transmitter of a message,
is also included.  Each increase in ``broadcast security'' resulted in a significant increase (\mbox{$P < 0.05$})
to the user's CL-privacy, as determined by
a one-way ANOVA test with Bonferroni's multiple comparison test on each incremental pair
(i.e., from a global key to link keys, and from link keys to
a non-broadcast environment).  The first increment illustrates the value to the adversary
of being able to learn more nodes along the communication path and the ultimate
destination of the communication, for the purposes of inferring from where the initial query could have come.  The
second increment illustrates the ultimate value to the adversary of overhearing anything at all.
On the other hand, the first increment in ``broadcast security'' resulted in a significant increase to the user's LOI-privacy
(\mbox{$P < 0.05$}), but the move from link keys to a non-broadcast environment did not (\mbox{$P > 0.10$}).
In other words, being able to decipher the ultimate destination of an overheard query that a compromised node
is not routing is helpful to an adversary in determining the user's location of interest, but an overheard query where
the adversary cannot decrypt the final destination is practically meaningless towards that goal.

\begin{figure*} 
\centering 
\begin{subfigure}[b]{0.48\textwidth}
\centering
\includegraphics[width=\textwidth]{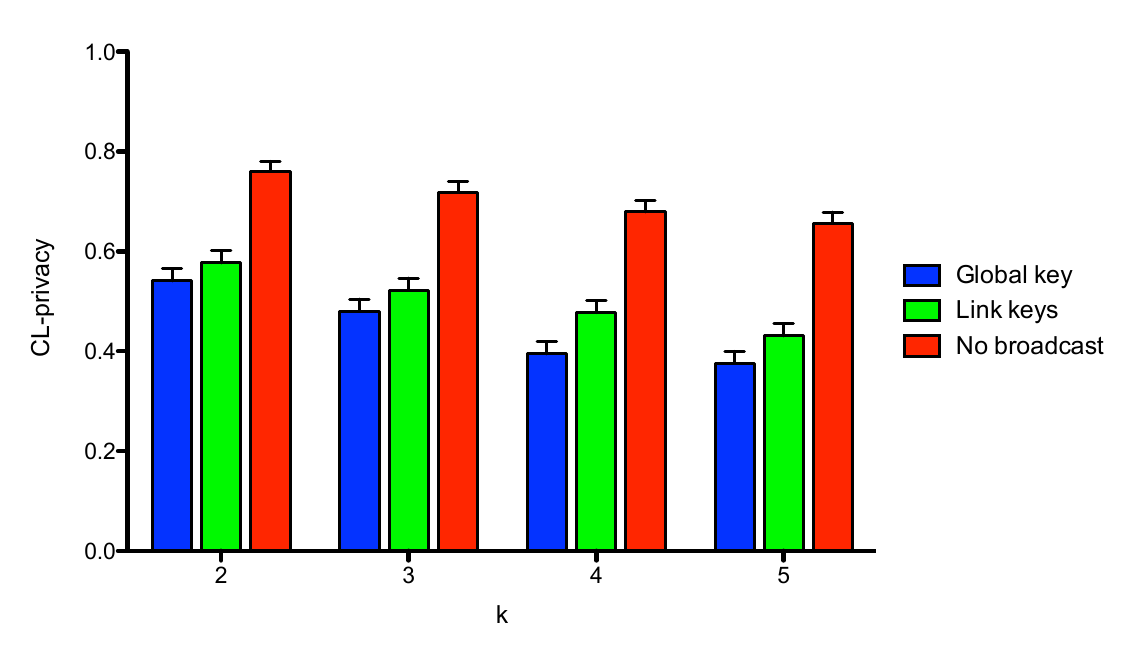}
\end{subfigure}
\hfill
\begin{subfigure}[b]{0.48\textwidth}
\centering
\includegraphics[width=\textwidth]{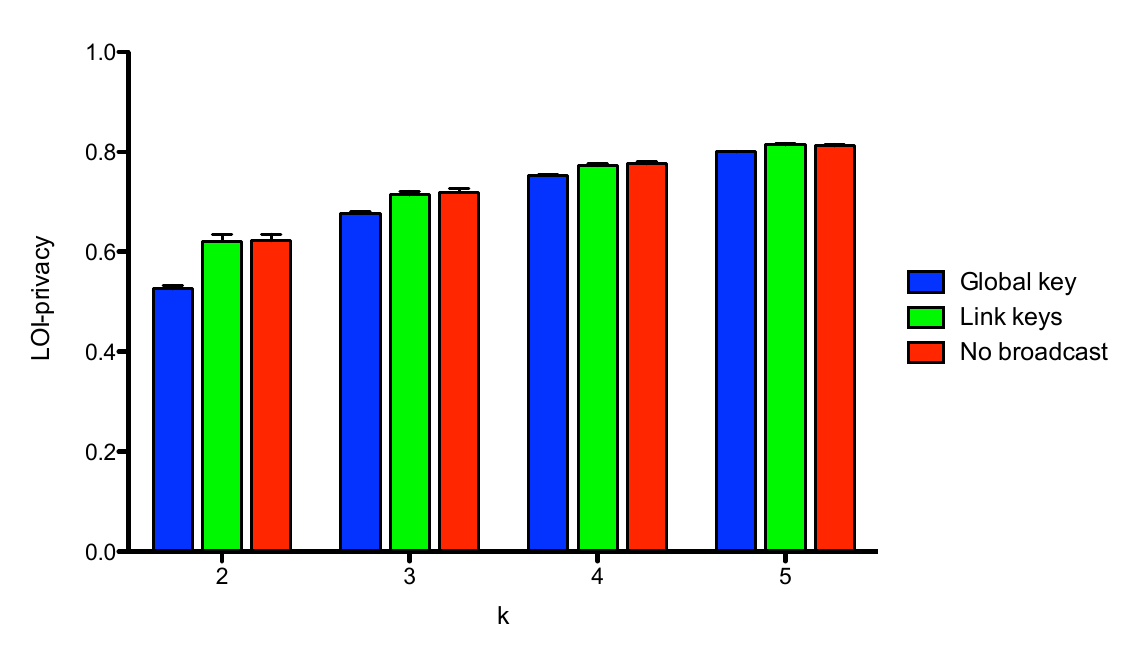}
\end{subfigure}
\caption{The user's CL-privacy and LOI-privacy, as the type of broadcast encryption varies, over 1000 trials.} 
\label{fig:static-adversary-results} 
\end{figure*}

Similar to the type of encryption used in the network, the type of routing used also may or may not be under
the control of the user. Here we compare the two types of routing described in Section~\ref{sect:ch2-cl}:
fixed shortest-path and random shortest-path routing.  In both cases, a query sent from the user's current location
to the query destination will arrive in the fewest number of hops possible; however, in fixed shortest-path routing,
a query sent between those two locations will always take the same path, whereas in random shortest path routing
it could potentially take many different paths.  The results are illustrated in Figure~\ref{fig:static-routing-results}.
Switching to random shortest-path routing significantly improved the user's CL-privacy, as the adversary could
not ascertain the origin of a query as easily from an overheard message.  Obviously, when using a global encryption
key, the type of routing used had no effect on the user's LOI-privacy.  However, the question remains: if link keys were
used instead of a global key, creating a situation in which the adversary could overhear query messages but not learn
their destination, would implementing random shortest-path routing increase the user's LOI-privacy by making it
more difficult for the adversary to narrow down possible destinations for overheard messages?  We repeated the
previous experiment, this time using link keys; the results are presented in Figure~\ref{fig:static-routing-link-results}.
Again, there was no distinction between the LOI-privacy results whether fixed shortest-path or random
shortest-path routing were used~---~in essence, being able to easily guess potential destinations
for overheard but non-decrypted packets is not highly important to the adversary.

\begin{figure*} 
\centering 
\begin{subfigure}[b]{0.48\textwidth}
\centering
\includegraphics[width=\textwidth]{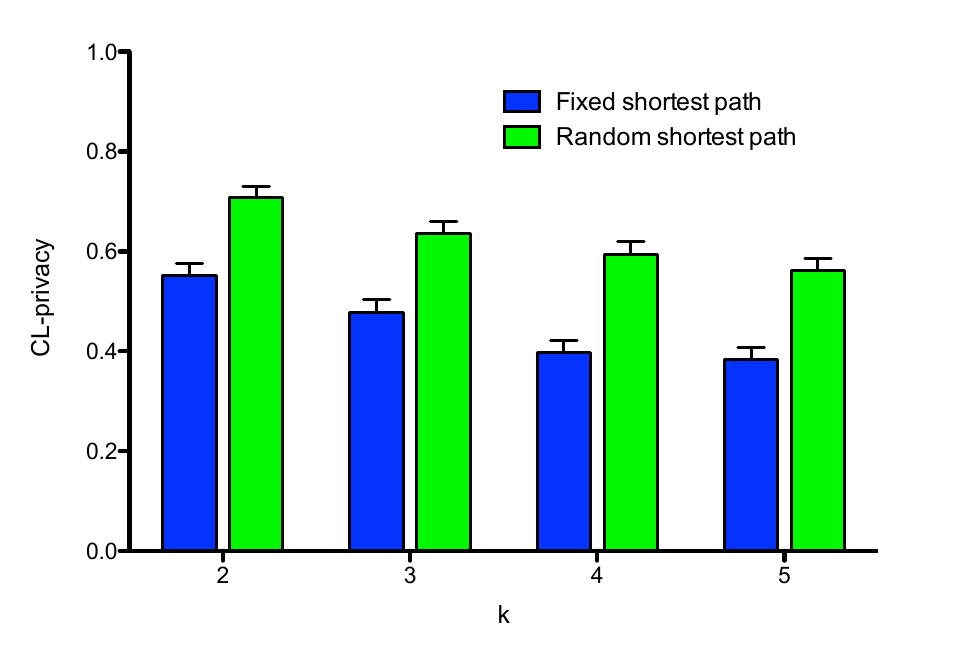}
\end{subfigure}
\hfill
\begin{subfigure}[b]{0.48\textwidth}
\centering
\includegraphics[width=\textwidth]{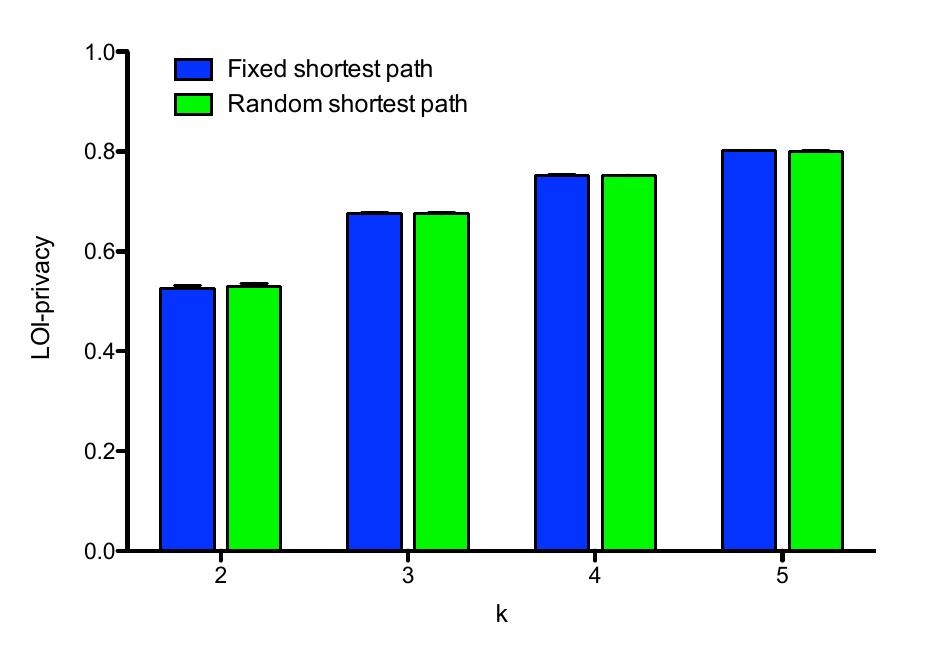}
\end{subfigure}
\caption{The user's CL-privacy and LOI-privacy, as the routing algorithm used in the network varies, over 1000 trials with a global key.} 
\label{fig:static-routing-results} 
\end{figure*}

\begin{figure*} 
\centering 
\begin{subfigure}[b]{0.48\textwidth}
\centering
\includegraphics[width=\textwidth]{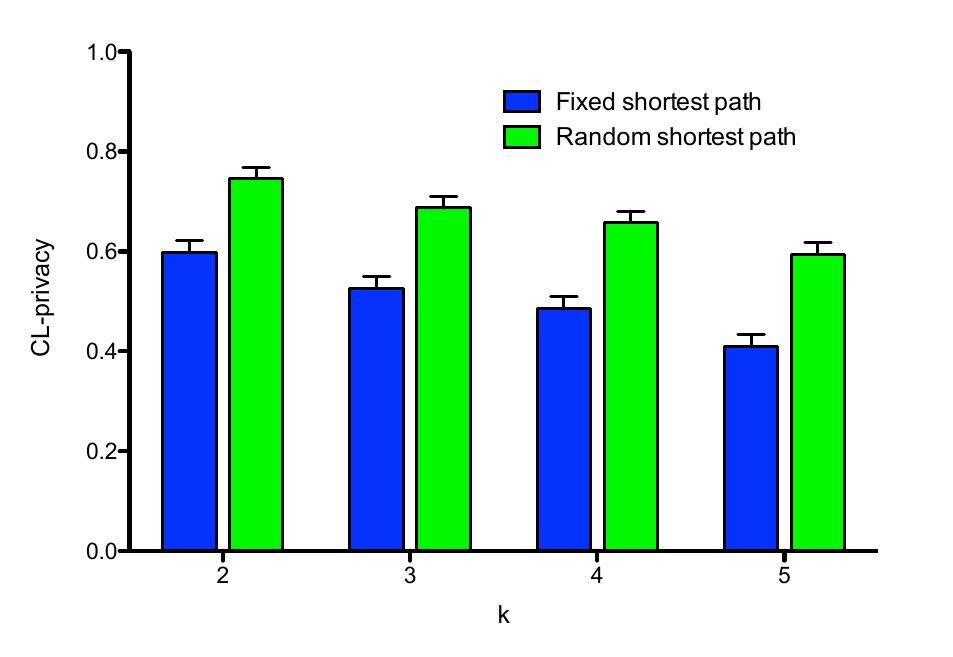}
\end{subfigure}
\hfill
\begin{subfigure}[b]{0.48\textwidth}
\centering
\includegraphics[width=\textwidth]{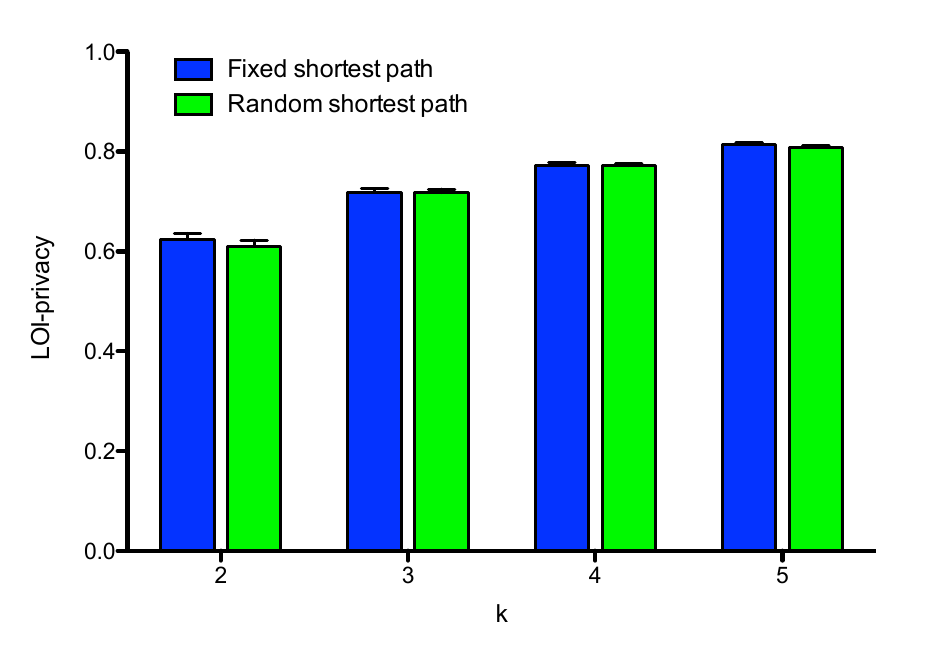}
\end{subfigure}
\caption{The user's CL-privacy and LOI-privacy, as the routing algorithm used in the network varies, over 1000 trials with link keys.} 
\label{fig:static-routing-link-results} 
\end{figure*}

\section{Experimental Results with Mobility}\label{sect:ch2-mobility}

Until now, we have run experiments in a uniform, featureless environment, in which for all nodes $N$ and $M$ in the network
\mbox{$\probC{N} = \probL{M}{N} = \frac{1}{|\mathcal{N}|}$}.
However, for a real network in an urban environment, some nodes will be located at more
population locations than other nodes.

To generate such a rich environment, we use DBS3 \cite{VOGT2012113} as the mobility simulator and the
maps of part of downtown Edmonton, Canada and of Fira, Greece shown in Figures \ref{fig:edmonton} and \ref{fig:fira}. The two maps are, respectively, a highly organized grid, and a more organic pattern  evolved from terrain constraints. %
We placed nodes centred on both maps in a grid pattern with 100~m spacing, yielding 112 nodes on the Edmonton map and
84 nodes on the Fira map.  On both maps, we initialized a 10-million-agent simulation, and computed the closest
node to each agent.  We then set
\begin{equation*}
\probC{N} = \displaystyle\frac{\nu_0(N)}{10000000}
\end{equation*}
for each \mbox{$N \in \mathcal{N}$}, where \mbox{$\nu_0(N)$} is the number of agents for whom $N$ is the
closest node at time zero.  We also set \mbox{$\probL{M}{N} = \probC{M}$} for each \mbox{$M \in \mathcal{N}$}.
In other words, the probability that a node is either the current location or the location of interest is proportional
to the amount of pedestrian traffic near it.
As before, the user's current location is chosen proportional to $\probC{N}$,
the user's real location of interest to be queried is chosen proportional to
$\probL{M}{N}$, and the \mbox{$k-1$} dummy queried locations are chosen uniformly at random. %
partitioning scheme).

\begin{figure}[htbp]
  \centering
  \begin{minipage}{0.48\textwidth}
    \centering
      \includegraphics[width=\linewidth]{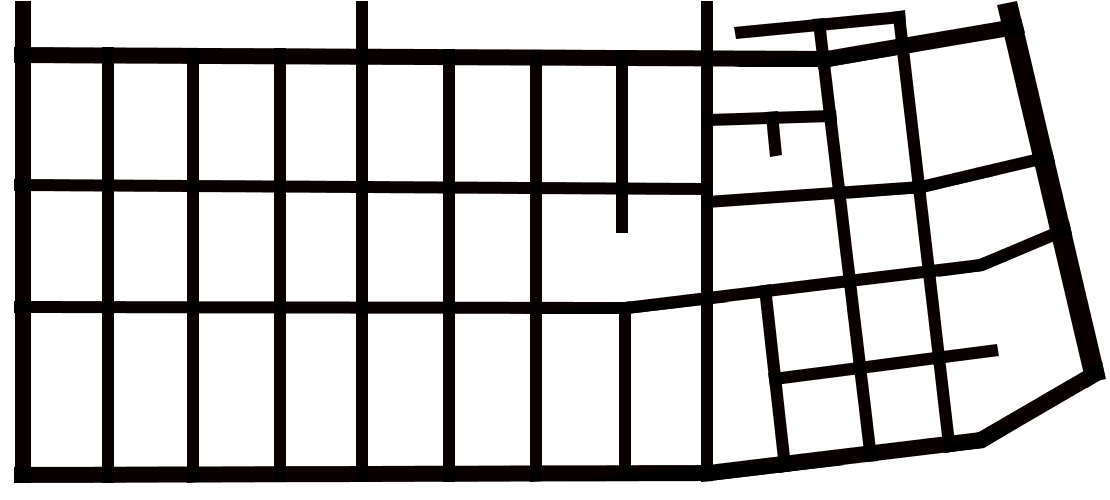}
      \caption{The map of downtown Edmonton.}\label{fig:edmonton}
  \end{minipage}
  \hfill
  \begin{minipage}{0.45\textwidth}
    \centering
  \includegraphics[width=\linewidth]{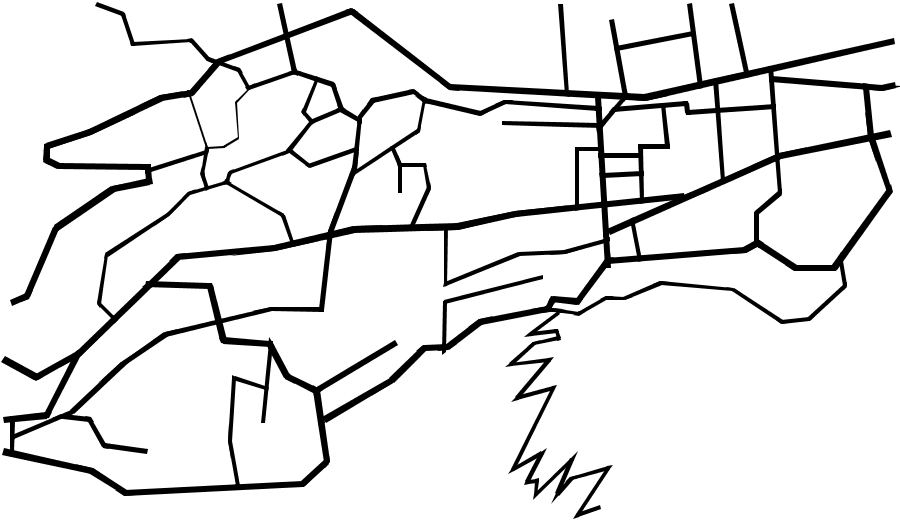}
  \caption{The map of Fira.}\label{fig:fira}
  \end{minipage}
\end{figure}

To determine how beneficial this statistical data is to an adversary, we compared the statistically rich environments
of Edmonton and Fira to otherwise equivalent environments in which
\mbox{$\probC{N} = \probL{M}{N} = \frac{1}{|\mathcal{N}|}$} for all nodes \mbox{$N, M \in \mathcal{N}$}.
As in the previous section,
message headers between nodes are encrypted using a global key that has been
compromised by the adversary.  The adversary has compromised 10 of the nodes in the network, and is using
them to eavesdrop on communications.  The user issues queries into the network from
their current location, and the adversary attempts to determine the user's current location and real location of
interest.  Because of the higher variance LOI-privacy results in the statistically rich environment, we increased
the number of trials from 1000 (used in the previous section) to 5000.

The results of this experiment are illustrated for the Edmonton map in Figure~\ref{fig:aggregate-edmonton}
and for the Fira map in Figure~\ref{fig:aggregate-fira}.  As evident from those graphs, both the user's CL- and
LOI-privacy fell considerably when this additional statistical data was available to the adversary.  Notably, LOI-privacy
was more seriously affected than CL-privacy by this additional information, because the random
partitioning scheme potentially resulted in low-probability LOIs being grouped with the user's real LOI.
Also, note that both privacy measures were affected more seriously in Fira, which is more complex from a
Space Syntax standpoint, and has more popular central thoroughfares.

\begin{figure*} 
\centering 
\begin{subfigure}[b]{0.48\textwidth}
\centering
\includegraphics[width=\textwidth]{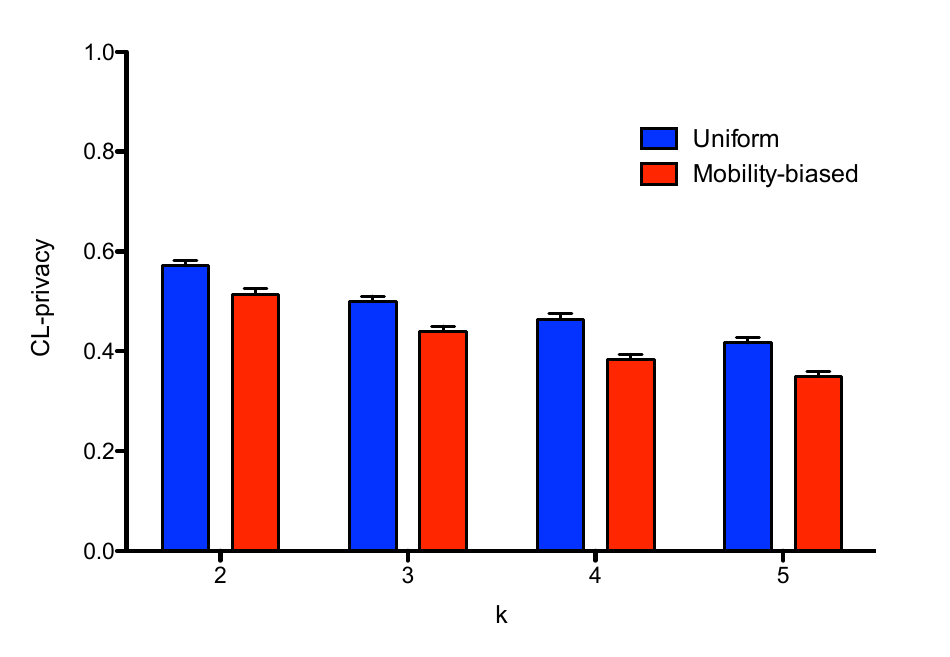}
\end{subfigure}
\hfill
\begin{subfigure}[b]{0.48\textwidth}
\centering
\includegraphics[width=\textwidth]{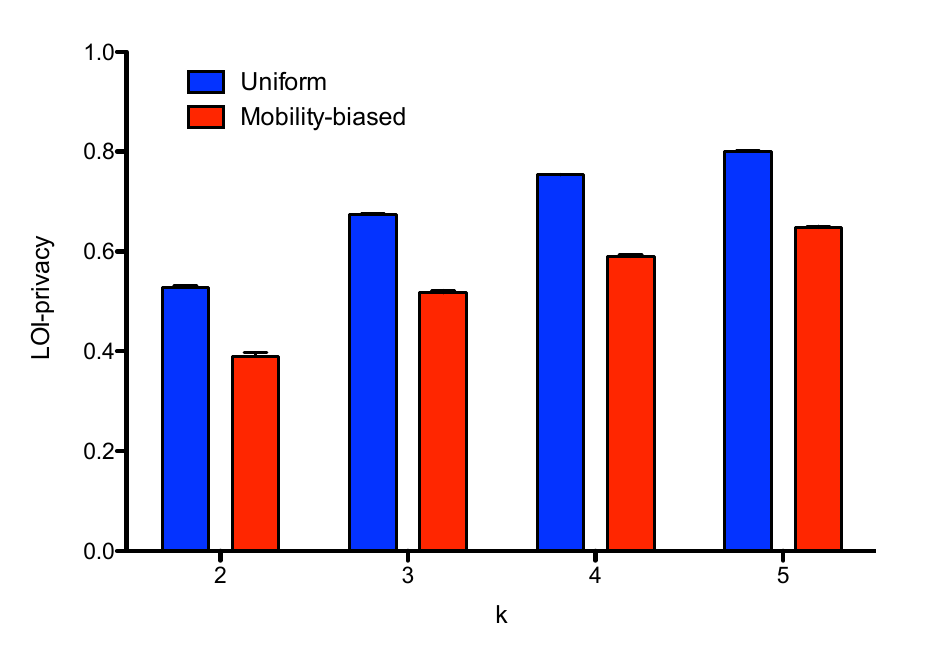}
\end{subfigure}
\caption{The user's CL-privacy and LOI-privacy, in statistically rich versus uniform environments, over 5000 trials in Edmonton.} 
\label{fig:aggregate-edmonton} 
\end{figure*}

\begin{figure*} 
\centering 
\begin{subfigure}[b]{0.48\textwidth}
\centering
\includegraphics[width=\textwidth]{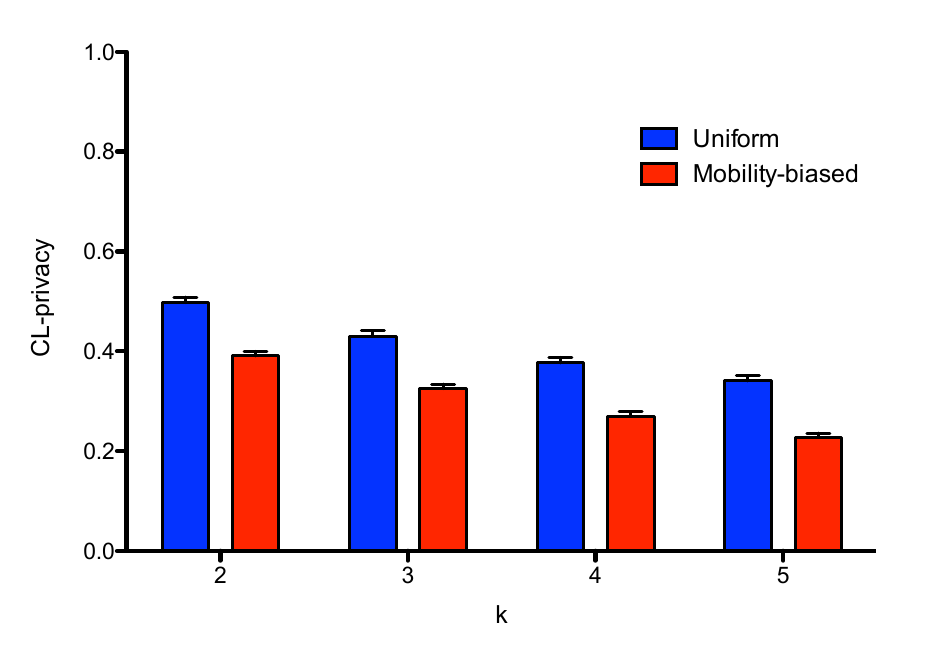}
\end{subfigure}
\hfill
\begin{subfigure}[b]{0.48\textwidth}
\centering
\includegraphics[width=\textwidth]{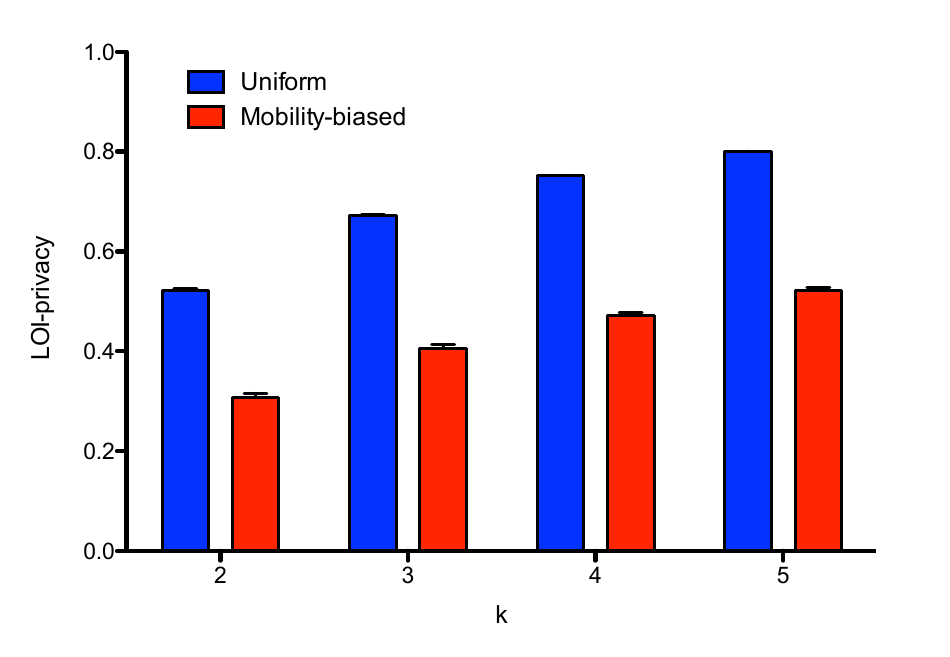}
\end{subfigure}
\caption{The user's CL-privacy and LOI-privacy, in the statistically rich versus uniform environments, over 5000 trials in Fira.} 
\label{fig:aggregate-fira} 
\end{figure*}

Another situation that could possibly benefit the adversary is if the adversary knows that the user is walking
from their current location to their location of interest in the near future.
In this way, the user's current location would affect the
probabilities that the various nodes in the network are the user's real location of interest, i.e.,
\probL{M}{N} for \mbox{$N, M \in \mathcal{N}$}.  To generate this even richer world, we advanced the
10-million-agent simulation to some time \mbox{$t \ge 0$}.  Keeping
\begin{equation*}
\probC{N} = \displaystyle\frac{\nu_0(N)}{10000000} \;,
\end{equation*}
we now set
\begin{equation*}
\probL{M}{N} = \displaystyle\frac{\nu_t(N, M)}{\nu_0(N)}
\end{equation*}
for each \mbox{$N, M \in \mathcal{N}$},
where \mbox{$\nu_t(N, M)$} is the number of agents for whom $N$ is the closest node at time zero
and $M$ is the closest node at time $t$.

To motivate that this additional short-time-horizon information could assist the
adversary, consider an adversary who, in the absence of any network communication, wants to predict the
user's location at some time \mbox{$t \ge 0$}.  Lacking any information about the user to utilize, the
adversary would just guess the most popular location on the map is the user's location, guessing successfully
with probability \mbox{$\max_{N \in \mathcal{N}}{\probC{N}}$}.
Compare that adversary to one who wants to predict the user's location at time \mbox{$t \ge 0$}, but also
has access to an oracle that reveals the user's location at time \mbox{$t = 0$}.
Such an adversary could successfully predict that user's location with probability
\begin{equation*}
\sum_{N \in \mathcal{N}}{\left[\probC{N} \cdot \max_{M \in \mathcal{N}}{\probL{M}{N}} \right]} \;.
\end{equation*}
These two adversaries are compared on the Edmonton map in Figure~\ref{fig:oracle-edmonton} and on
the Fira map in Figure~\ref{fig:oracle-fira}.  The more complex Fira map with few thoroughfares leads to a higher
baseline success rate for the adversary without an oracle than in Edmonton,
and the natural inclination of pedestrians to move towards
and stay on those thoroughfares results in the oracle-possessing adversary's success rate declining towards
the baseline more quickly than in Edmonton.  However, in either case, the adversary with the oracle has a clear
advantage over the adversary without the oracle up to a time threshold of \mbox{$t \approx$ 10--15}
(where $t$ is measured in minutes).

\begin{figure}[htbp]
  \centering
  \begin{minipage}{0.48\textwidth}
    \centering
    \includegraphics[width=\linewidth]{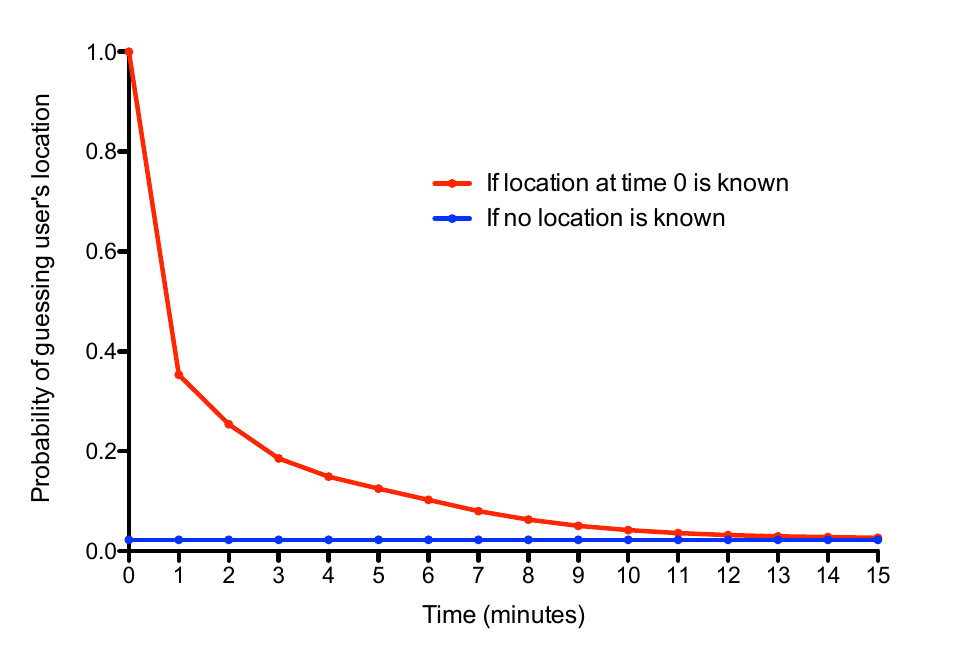}
    \caption{The success rate of an adversary predicting a user's location at time \mbox{$t \ge 0$} with and without an oracle revealing the user's location at time \mbox{$t = 0$}, in Edmonton.}\label{fig:oracle-edmonton}
  \end{minipage}
  \hfill
  \begin{minipage}{0.48\textwidth}
    \centering
    \includegraphics[width=\linewidth]{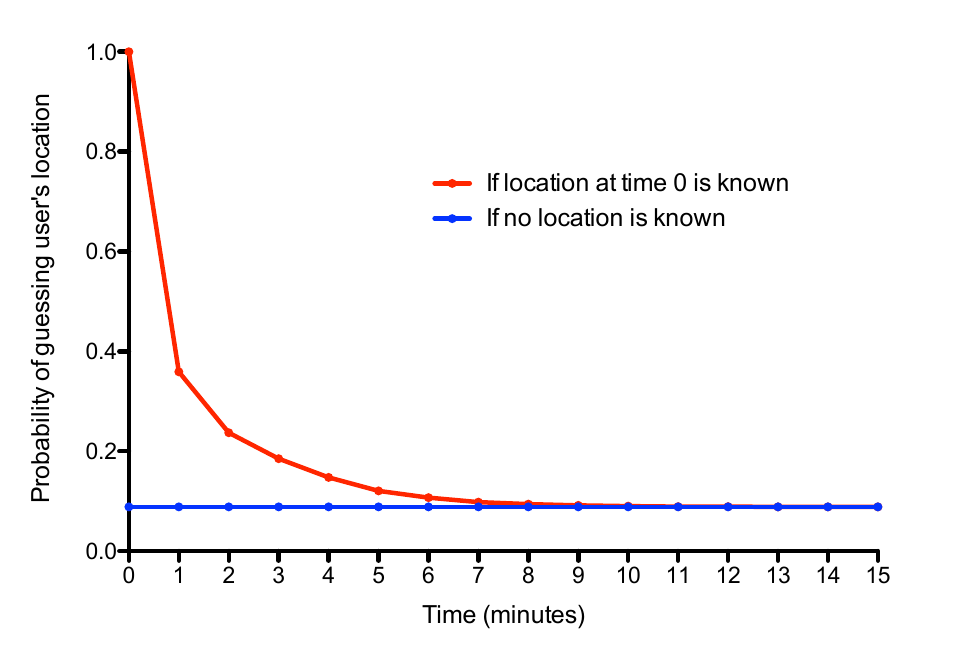}
    \caption{The success rate of an adversary predicting a user's location at time \mbox{$t \ge 0$} with and without an oracle revealing the user's location at time \mbox{$t = 0$}, in Fira.}\label{fig:oracle-fira}
  \end{minipage}
\end{figure}

To further investigate this concept of a time threshold, we reran the previous
$k$-query experiment on the Edmonton and Fira maps.  This time, however, we assumed that
the user is travelling from their current location to their location of interest~---~i.e., we used the
richer definition of \probL{M}{N}, for some chosen time threshold $t$ between when the user sends their
$k$ queries and when they arrive at their real location of interest.
The resulting CL- and LOI-privacy for the user in Edmonton, over varying
time thresholds, is
illustrated in Figure~\ref{fig:mobile-edmonton}; the results of the
same experiment in Fira are in Figure~\ref{fig:mobile-fira}.

\begin{figure*} 
\centering 
\begin{subfigure}[b]{0.48\textwidth}
\centering
\includegraphics[width=\textwidth]{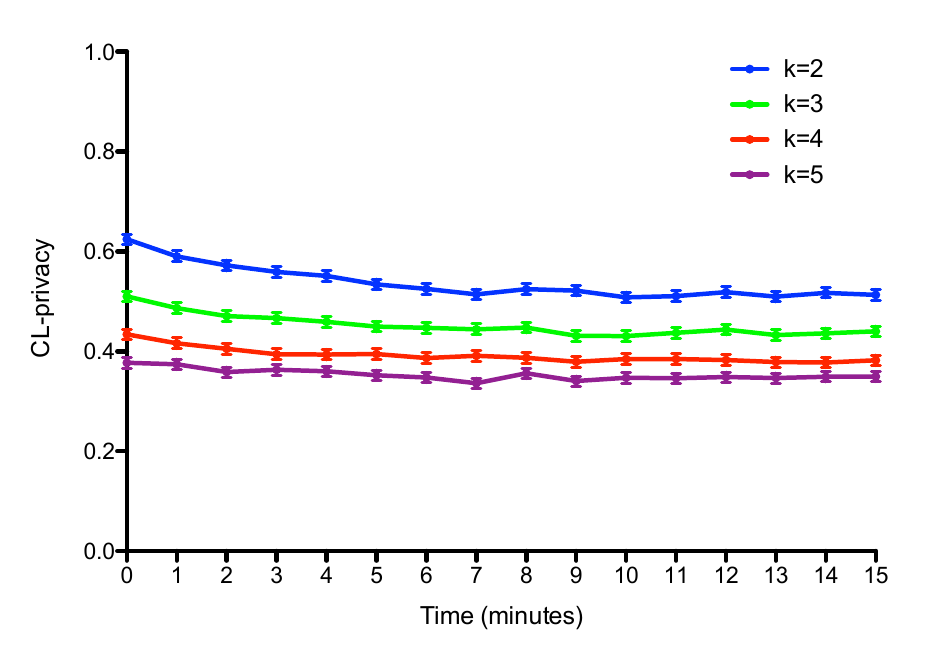}
\end{subfigure}
\hfill
\begin{subfigure}[b]{0.48\textwidth}
\centering
\includegraphics[width=\textwidth]{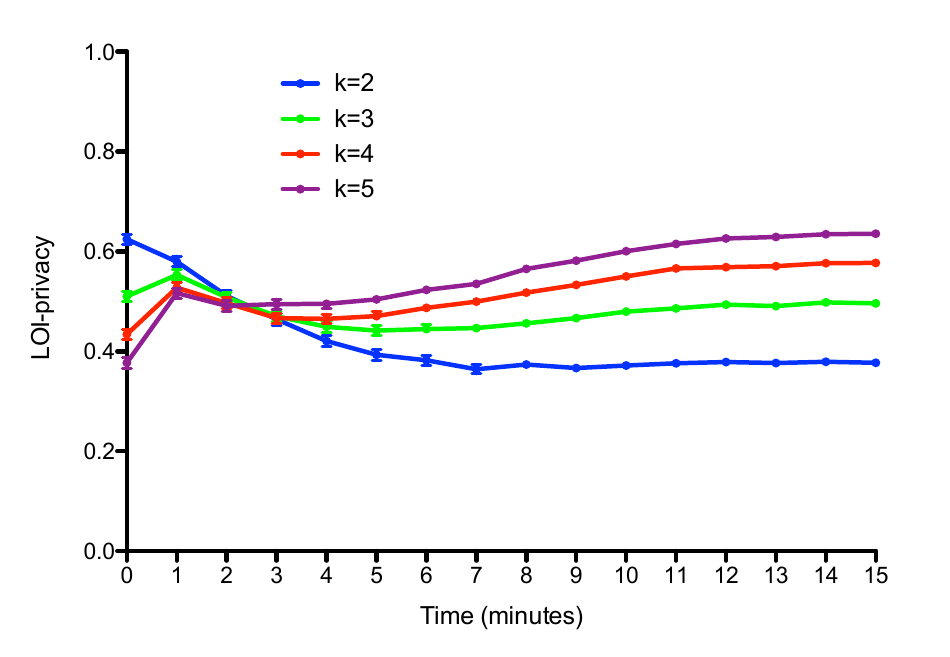}
\end{subfigure}
\caption{The user's CL-privacy and LOI-privacy, when the user queries their location of interest at varying times before arriving at it, over 5000 trials in Edmonton.} 
\label{fig:mobile-edmonton} 
\end{figure*}

\begin{figure*} 
\centering 
\begin{subfigure}[b]{0.48\textwidth}
\centering
\includegraphics[width=\textwidth]{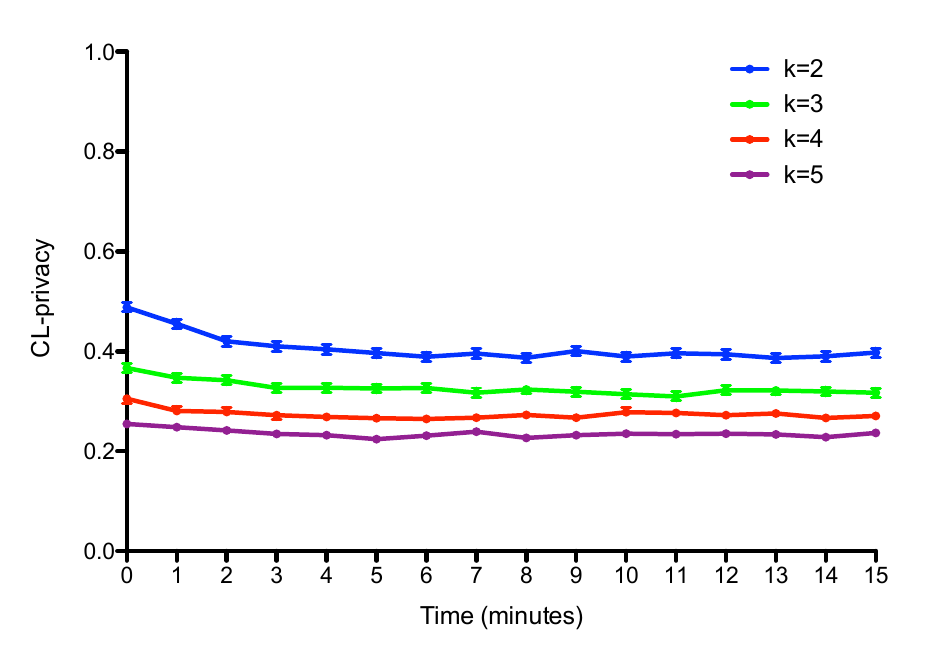}
\end{subfigure}
\hfill
\begin{subfigure}[b]{0.48\textwidth}
\centering
\includegraphics[width=\textwidth]{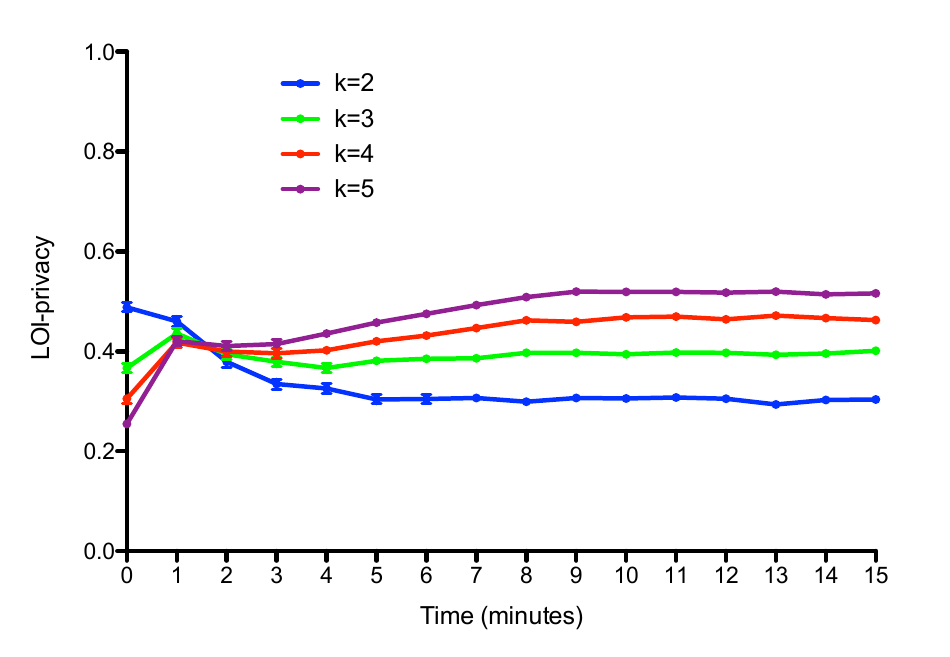}
\end{subfigure}
\caption{The user's CL-privacy and LOI-privacy, when the user queries their location of interest at varying times before arriving at it, over 5000 trials in Fira.} 
\label{fig:mobile-fira} 
\end{figure*}

Counter-intuitively, for small $t$ values, the user's CL-privacy is higher than at larger $t$ values.
Consider the case where, for \mbox{$t=0$},
one of the queried locations is guaranteed to be the user's current location,
effectively decrementing $k$ from the point of view of the adversary's ability to intercept
the user's messages.  Assuming that the user's CL-privacy has plateaued by \mbox{$t = 15$}
(based on the oracle experiment), we performed a
one-way ANOVA test with Bonferroni's multiple comparison test between the user's CL-privacy
at \mbox{$t = 15$} and each previous time (i.e., comparing
\mbox{$t = 0$} to \mbox{$t = 15$}, \mbox{$t = 1$} to \mbox{$t = 15$}, etc.).
In Edmonton using \mbox{$k=3$}, there was no significant difference (\mbox{$P > 0.05$})
in the CL-privacy at any time \mbox{$4 \le t \le 14$} and the CL-privacy at \mbox{$t = 15$}.
Similarly, in Fira using \mbox{$k=3$}, there was no significant difference (\mbox{$P > 0.05$})
in the CL-privacy at time \mbox{$3 \le t \le 14$} and the CL-privacy at \mbox{$t = 15$}.
Based on these results, we conclude that any benefit to the user's CL-privacy from a small $t$ threshold
from query to arrival has vanished by 4 and 3~minutes respectively on the Edmonton and Fira maps
when \mbox{$k=3$}.

Regarding the LOI-privacy results,
recall from Section~\ref{sect:ch2-results} the fundamental trade-off between the privacy of the user's current
location and that of their location of interest: as $k$ increases, the user's LOI-privacy increases at the cost of
their CL-privacy.  However, recall Lemma~\ref{lemma:equal}, which states that when the adversary knows the user's
current location is their real location of interest (e.g., in this experiment when \mbox{$t = 0$}),
the resulting CL- and LOI-privacies are identical.
This apparent paradox is visible in the LOI-privacy graphs wherein the lines representing the various $k$ values
cross as $t$ increases, before plateauing.  These graphs tell us that, when the adversary knows that the user's
real LOI is at or near their current location,
the LOI-privacy metric is really measuring the privacy of the user's current location.
Similar to the ANOVA and Bonferroni results on CL-privacy, there was no significant difference (\mbox{$P > 0.05$})
between the user's LOI-privacy at any time \mbox{$10 \le t \le 14$} and at time \mbox{$t = 15$} in Edmonton
using \mbox{$k=3$}; the same holds true in Fira using \mbox{$k=3$} for any time \mbox{$6 \le t \le 14$}.

From these results, we can conclude that if a user is travelling from their current location to their real location
of interest, they should send their $k$ queries long enough in advance of embarking
to mitigate any advantage the adversary has in predicting the user's real LOI.  How long is ``long enough'' will
depend on the map in question and the number of queries sent.  However, ``long enough'' in our experiments
was a surprisingly short period of time for movement at pedestrian mobility speeds, e.g., a mere
\mbox{$t = 10$ minutes} when using \mbox{$k=3$} on the Edmonton map.

\section{Conclusions}\label{sect:ch2-conclusion}

In this paper, we presented a user-controlled methodology for preserving the privacy of the
user's location of interest in a wireless multi-hop network: issuing $k$ queries instead of just one.
This methodology does not require any
special security measures in the network itself (as opposed to, e.g., onion routing).  There
is a direct trade-off between preserving the privacy of the user's location of interest and the privacy
of their current location, controlled by the user's choice of the security constant $k$, i.e., the number
of queries sent by the user.

The presence of user mobility data, reflecting the popularity of different locations, gives the adversary
a significant advantage in determining a user's current location and location of interest, relative to an
adversary operating in a featureless, uniform environment.  It is even more difficult for users
moving from their current location
to their location of interest to preserve their privacy;
however, those users can mitigate this advantage of the adversary by leaving a long enough time
between their query and their arrival at the location of interest.

Future work on this scheme should seek network partitioning schemes that yield guarantees of
area-of-interest privacy and/or location-of-interest privacy while still satisfying the closure requirement.
In particular, satisfying the closure condition in the presence of unequal location popularities
will be a substantial challenge.

\begin{acks}
This work was supported by the Natural Sciences and Engineering Research Council of Canada (NSERC) Discovery Grant program. Ryan Vogt was also supported by an Alberta Innovates / iCore scholarship. 
\end{acks}

\bibliographystyle{ACM-Reference-Format}
\bibliography{kanon,mobility}

\appendix

\section{NP-completeness Proofs} \label{hardness}

\begin{theorem}\label{thm:partition-graph}
{\sc Partition-Graph} is NP-complete.
\end{theorem}
\begin{proof}
First, we show that {\sc Partition-Graph} is in NP.  Consider a candidate solution.
The distance between any two vertices in the candidate solution could be established by, e.g.,
\mbox{$\mathcal{O}(k)$} breadth-first searches, each running in polynomial time and space.
The candidate solution could then be verified by checking the distance between all pairs of
vertices in the solution, requiring \mbox{$\mathcal{O}(k^2)$} distance lookups.

Next, we perform the Turing reduction.
Let \mbox{$\proc{A}(G, k, d)$} be an algorithm that solves {\sc Partition-Graph} (i.e., returns
either \const{True} or \const{False}).
Now, consider the following problem:

\noindent
\item {\sc Independent-Set}: You are given a graph \mbox{$G = \langle V, E\rangle$},
\mbox{$k \in \mathbb{Z} \le |V|$}.
Is there a subset \mbox{$S \subseteq V$} with \mbox{$|S| \ge k$}
such that no two vertices in $S$ are joined by an edge in $E$?

{\sc Independent-Set} is known to be NP-complete~\cite{karp:reducibility}.
Note that an independent set of size $k$ is equivalent to a partition of size $k$ with
a distance of at least two hops between all pairs of vertices in the partition.  Therefore,
an instance \mbox{$\langle G, k\rangle$} of {\sc Independent-Set} can be solved with a single
call to the {\sc Partition-Graph} solver, \mbox{$\proc{A}(G, k, 2)$}.
\end{proof}

\begin{theorem}
{\sc Partitioning-Graph} is NP-complete.
\end{theorem}
\begin{proof}
First, we show that {\sc Partitioning-Graph} is in NP.  Consider a candidate solution.
The distance between any two vertices in the graph could be established by, e.g.,
\mbox{$\mathcal{O}(|V|)$} breadth-first searches, each running in polynomial time and space.
The candidate solution could then be verified by checking the distance between all pairs of
vertices in the same partition in the solution, requiring \mbox{$\mathcal{O}(|V|^2)$} distance lookups.

Next, we perform the Turing reduction.
Let \mbox{$\proc{A}(G, k_\mathit{min}, k_\mathit{max}, d)$} be an algorithm that solves {\sc Partitioning-Graph}
(i.e., returns either \const{True} or \const{False}).
Now, consider the following problem:

\noindent
{\sc Partition-Into-Triangles}: You are given a graph \mbox{$G = \langle V, E\rangle$} where
\mbox{$|V| = 3q$} for some \mbox{$q \in \mathbb{Z}$}.
Can $V$ be partitioned into disjoint sets \mbox{$V_1, \ldots, V_q$}, where \mbox{$|V_i| = 3$} and
$V_i$ forms a clique, for all $i$?

{\sc Partition-Into-Triangles} is known to be NP-complete~\cite{schaefer:triangles,garey:computers}.
Given an instance \mbox{$\langle G \rangle$} of {\sc Partition-Into-Triangles}, first construct the
complement graph, \mbox{$G^\prime = \langle V, E^\prime \rangle$}, in which two vertices are connected
by an edge if and only if they were not connected by an edge in $G$.  Note that $V$ on $G$ can be partitioned
into disjoint sets $V_i$ with \mbox{$|V_i| = 3$} that all form cliques, if and only if $V$ on $G^\prime$ can be
partitioned into disjoint sets $V_i$ with \mbox{$|V_i| = 3$} that all form independent sets.
Having no edges connecting the three nodes is equivalent to them being a partition of size three
with a distance of at least two hops between all pairs of vertices in the partition.
Therefore, calling the {\sc Partitioning-Graph} solver a single time,
\mbox{$\proc{A}(G^\prime, 3, 3, 2)$},
yields the solution to a {\sc Partition-Into-Triangles} instance \mbox{$\langle G \rangle$}.
\end{proof}

\section{Proofs of Lemmas and Theorems}\label{longproofs}
 
\begin{proof}[Proof of Lemma \ref{lemma:set-union}]
If there is no way to choose a unique integer from each of \mbox{$\mathcal{S}_1, \ldots, \mathcal{S}_m$}
(including if \mbox{$\mathcal{S}_i = \emptyset$} for any $i$), then
\mbox{$\mathcal{G}_m^\prime = \emptyset$}, and we are done.
Assume there is a way to choose $m$ unique integers,
one from each of the $\mathcal{S}_i$ (which implies that \mbox{$\mathcal{S}_i \neq \emptyset$} for all $i$).
The proof will be by induction on \mbox{$m \in \mathbb{Z} \ge 1$}.

If \mbox{$m = 1$}, then \mbox{$\mathcal{G}_m^\prime = \left\{\left\{s\right\}~|~s \in \mathcal{S}_1\right\}$},
because \mbox{$\{s\}$} is a sampling if and only if \mbox{$s \in \mathcal{S}_1$}.
As such,
\begin{equation*}
\bigcup_{\mathcal{G}_m \in \mathcal{G}_m^\prime}\mathcal{G}_m
= \bigcup_{s \in \mathcal{S}_1}\left\{s\right\}
= \mathcal{S}_1
= \bigcup_{1 \le i \le m}\mathcal{S}_i\;.
\end{equation*}

Assume for \mbox{$j \in \mathbb{Z} \ge 1$},
\begin{equation*}
\bigcup_{\mathcal{G}_j \in \mathcal{G}_j^\prime}\mathcal{G}_j = \bigcup_{1 \le i \le j}\mathcal{S}_i\;,
\end{equation*}
and let \mbox{$\mathcal{S}_{j+1} \subseteq \mathbb{Z}$} be another set of integers.
If it is not possible to choose \mbox{$j+1$} unique integers, one from each of
\mbox{$\mathcal{S}_1, \ldots, \mathcal{S}_{j+1}$}, then
\mbox{$\mathcal{G}_{j+1}^\prime = \emptyset$} and we are done.
Assume it is possible to choose \mbox{$j+1$} unique integers as described above.
If there are any elements \mbox{$e \in \mathcal{S}_{j+1}$} that are not contained in
\mbox{$\cup_{1 \le i \le j}\mathcal{S}_i$}, then a sampling \mbox{$\mathcal{G}_{j+1}$} could be constructed as
\mbox{$\mathcal{G}_{j+1} = \mathcal{G}_j \cup \{e\}$}, for any sampling \mbox{$\mathcal{G}_j$}
on \mbox{$\mathcal{S}_1, \ldots, \mathcal{S}_j$}.
Hence,
\begin{equation*}
\bigcup_{\mathcal{G}_{j+1} \in \mathcal{G}_{j+1}^\prime}\mathcal{G}_{j+1} = \bigcup_{1 \le i \le j+1}\mathcal{S}_i\;,
\end{equation*}
because there is guaranteed to be a sampling \mbox{$\mathcal{G}_{j+1}$} that contains any
element of \mbox{$\mathcal{S}_{j+1}$} not contained in \mbox{$\cup_{1 \le i \le j}\mathcal{S}_i$}.
\end{proof}
\begin{proof}[Proof of Lemma \ref{lemma:equal}]
First, note that if $\probL{N}{M}$ is as defined, then %
\begin{align*}
\phi(N, M) &=\begin{cases}
\displaystyle\frac{\probL{M}{N} \cdot \chi(N, M)}
{\sum_{P \in \mathcal{L}_N}{\probL{P}{N} \cdot \chi(N, P)}}
& \mbox{if $M \in \mathcal{L}_N$} \\
0 & \mbox{$if M \not\in \mathcal{L}_N$}\end{cases} \\
&=\begin{cases}
\displaystyle\frac{1 \cdot \chi(N, M)}{1 \cdot \chi(N, M) + \sum_{P \in \mathcal{L}_N \setminus \{M\}}{0 \cdot \chi(N, P)}} & \mbox{if $M \in \mathcal{L}_N$ and $N = M$} \\
0 & \mbox{otherwise}\end{cases} \\
&=\begin{cases}
1 & \mbox{if $M \in \mathcal{L}_N$ and $N = M$} \\
0 & \mbox{otherwise.}\end{cases}
\end{align*}
We can further simplify $\phi$ by noting that if \mbox{$N = M$}, then under our assumption about
$\probL{N}{M}$ (which implies one of the $k$ queried locations will always be the user's current location),
it is guaranteed that \mbox{$M \in \mathcal{L}_N$}.  If the adversary has compromised the
user's current location $C$, the adversary would process the query directed to the real location of interest (i.e., $C$)
and add $C$ to \mbox{$\mathcal{L}_C$}.  If the adversary has not compromised the user's current location,
then the adversary would not overhear all of the user's $k$ queries, and would, for any possible current location
$N$, add $N$ to \mbox{$\mathcal{L}_N$},
since a query to a non-compromised current location always represents a possible
assignment of a destination to a query with an unknown destination.  As such,
\begin{equation*}
\phi(N,M) =\begin{cases}
1 & \mbox{if $N = M$} \\
0 & \mbox{otherwise.}\end{cases}
\end{equation*}

Using the simplified form of $\phi$, we can simplify $\omega$ as well:
\begin{align*}
\omega\left(M\right) &= \sum_{N \in \mathcal{C}}{\psi\left(N,M\right)} \\
&= \sum_{N \in \mathcal{C}}{\probpC{N}\cdot\phi\left(N,M\right)} \\
&=\begin{cases}
\probpC{M} \cdot 1 + \sum_{N \in \mathcal{C} \setminus \{M\}}{\probpC{N} \cdot 0} & \mbox{if $M \in \mathcal{C}$} \\
\sum_{N \in \mathcal{C}}{\probpC{N}\cdot 0} & \mbox{if $M \not\in \mathcal{C}$}\end{cases} \\
&=\begin{cases}
\probpC{M} & \mbox{if $M \in \mathcal{C}$} \\
0 & \mbox{if $M \not\in\mathcal{C}$\;.}\end{cases}
\end{align*}

Next, observe that \mbox{$\mathcal{C} \subseteq \cup_{N \in \mathcal{C}}{\mathcal{L}_N}$}, by our earlier
observation that \mbox{$N \in \mathcal{L}_N$} for all \mbox{$N \in \mathcal{C}$}.
Therefore,
\begin{equation*}
\max_{M \in \cup_{N \in \mathcal{C}}{\mathcal{L}_N}}{\omega\left(M\right)} = \max_{M \in \mathcal{C}}{\omega\left(M\right)} \;,
\end{equation*}
since \mbox{$\omega(M) = 0$} for any \mbox{$M \in \cup_{N \in \mathcal{C}}{\mathcal{L}_N} \setminus \mathcal{C}$}.
Hence,
\begin{align*}
\proc{M}_\mathit{LOI} &= 1 - \omega\left(L^*\right) = 1 - \max_{M \in \cup_{N \in \mathcal{C}}{\mathcal{L}_N}}{\omega\left(M\right)} \\
&= 1 - \max_{M \in \mathcal{C}}{\omega\left(M\right)} = 1 - \max_{M \in \mathcal{C}}{\probpC{M}} \\
&= \proc{M}_\mathit{CL} \;,
\end{align*}
demonstrating that the two measures are equivalent when the adversary knows that the user's current location is
their real location of interest.
\end{proof}

\begin{proof}[Proof of Lemma \ref{lemma:compromised}]
If the node at the user's current location $C$ is compromised, the adversary will ``guess'' with certainty
that the user's current location is \mbox{$C^* = C$}~---~i.e., \mbox{$\mathcal{C} = \{C\}$}.  The adversary
will also know all of the locations queried, i.e., \mbox{$\mathcal{L}_C = \mathcal{L}$}.
In this case, the user's CL-privacy drops to
\begin{align*}
\proc{M}_\mathit{CL} &= 1 - \probpC{C^*} \\
&= 1-1 = 0 \;.
\end{align*}
The user's LOI-privacy also drops to
\begin{align*}
\proc{M}_\mathit{LOI} &= 1 - \omega\left(L^*\right) \\
&= 1 - \max_{M \in \cup_{N \in \mathcal{C}}{\mathcal{L}_N}}{\omega\left(M\right)} = 1 - \max_{M \in \mathcal{L}}{\omega\left(M\right)} = 1 - \max_{M \in \mathcal{L}}{\left[\sum_{N \in \mathcal{C}}{\psi\left(N, M\right)}\right]} \\
&= 1 - \max_{M \in \mathcal{L}}{\psi\left(C^*, M\right)} = 1 - \max_{M \in \mathcal{L}}{\left[ \probpC{C^*} \cdot \phi\left(C^*, M \right)\right]} = 1 - \max_{M \in \mathcal{L}}{\phi\left(C^*, M\right)} \\
&= 1 - \max_{M \in \mathcal{L}}{\left[\displaystyle\frac{\probL{M}{C^*} \cdot \chi(C^*, M)}{\sum_{P \in \mathcal{L}}{\probL{P}{C^*} \cdot \chi(C^*, P)}}\right]} \;.
\end{align*}
Note that, with the user's current location compromised, \mbox{$M \in \mathcal{K}$} for all
\mbox{$M \in \mathcal{L}$}.  Therefore, \mbox{$\chi(C^*, M) = 1$} for all \mbox{$M \in \mathcal{L}$}.  Hence,
\begin{align*}
\proc{M}_\mathit{LOI}
&= 1 - \max_{M \in \mathcal{L}}{\left[\displaystyle\frac{\probL{M}{C^*} \cdot \chi(C^*, M)}{\sum_{P \in \mathcal{L}}{\probL{P}{C^*} \cdot \chi(C^*, P)}}\right]} \\
&= 1 - \max_{M \in \mathcal{L}}{\left[\displaystyle\frac{\probL{M}{C^*}}{\sum_{P \in \mathcal{L}}{\probL{P}{C^*}}}\right]} \\
&= 1 - \max_{M \in \mathcal{L}}{\probGlobal\left(L = M~|~C = C^*~\text{and}~L \in \mathcal{L}\right)} \;.
\end{align*}
In a uniform, featureless environment, where from any current location the probability of any location
being any user's real LOI is $\frac{1}{n}$, that computation simplifies to
\begin{align*}
\proc{M}_\mathit{LOI} &= 1 - \max_{M \in \mathcal{L}}{\left[\displaystyle\frac{\probL{M}{C^*}}{\sum_{P \in \mathcal{L}}{\probL{P}{C^*}}} \right]} \\
&= 1 - \max_{M \in \mathcal{L}}{\left[\displaystyle\frac{\frac{1}{n}}{\sum_{P \in \mathcal{L}}{\frac{1}{n}}} \right]} = 1 - \max_{M \in \mathcal{L}}{\left[\displaystyle\frac{\frac{1}{n}}{k \cdot \frac{1}{n}} \right]} \\
&= 1 - \frac{1}{k} = \frac{k-1}{k} \;,
\end{align*}
thus completing the proof.
\end{proof}

\begin{proof}[Proof of Theorem \ref{theorem:2approx}]
Expanding on the notation from Lemma~\ref{lemma:greedy-partition}, let $Y$ be the optimal
result to an instance of the partition problem with minimal pairwise distance $d$.  Let $X$ be
the result from the greedy algorithm with minimal pairwise distance $d^\prime$, and assume
\mbox{$d^\prime < \frac{d}{2}$}.

Define, for all \mbox{$y \in Y$}, \mbox{$\proc{C}_{s}(y)$} as the single closest node in
$X_s$ to $y$ (with ties in distance being broken by the index of the node in $X_s$).  Formally,
\begin{equation*}
\proc{C}_{s}(y) = \argmin_i\{x_i \in X_s~|~\proc{D}(x_i, y)~\text{is minimized\}}\;.
\end{equation*}
By definition, \mbox{$\left|X_{s-1}\right| < k = \left|Y\right|$}.  Therefore, by the pigeon-hole principle,
there must exist \mbox{$y_a, y_b \in Y, a \neq b$}, such that
\mbox{$\proc{C}_{s-1}(y_a) = \proc{C}_{s-1}(y_b)$} (which we denote as $x_c$).

Since $x_c$ is the closest node in $X_{s-1}$ to $y_a$, and we know from
Lemma~\ref{lemma:greedy-partition} that there must exist some node \mbox{$x \in X_{s-1}$}
such that \mbox{$\proc{D}(y_a, x) \le d^\prime$}, we know that \mbox{$\proc{D}(y_a, x_c) \le d^\prime$}.
Similarly, \mbox{$\proc{D}(y_b, x_c) \le d^\prime$}.  Therefore, by the triangle inequality,
\begin{equation*}
\proc{D}(y_a, y_b) \le 2 \cdot d^\prime < 2 \cdot \frac{d}{2} = d \;,
\end{equation*}
meaning that the minimal pairwise distance in $Y$ is less than $d$, which is a contradiction.
\end{proof}

\end{document}